\documentclass[fleqn,10pt]{wlscirep}
\usepackage[utf8]{inputenc}
\usepackage[T1]{fontenc}

\usepackage[final]{changes}

\definechangesauthor[name={v1small}, color=orange]{v1}

\newtheorem{theorem}{Theorem}
\newtheorem{definition}{Definition}
\newtheorem{proof}{Proof}

\addto\captionsenglish{}
\addto\captionsenglish{}

\title{The effect of distant connections on node anonymity in complex networks}

\author[1,2,*]{Rachel G. de Jong}
\author[2, 1]{Mark P. J. van der Loo}
\author[1]{Frank W. Takes}
\affil[1]{Leiden University, LIACS, 2333 CA Leiden, The Netherlands} 
\affil[2]{Statistics Netherlands, Research and Development, 2492JP The Hague, The Netherlands}
\affil[*]{r.g.de.jong@liacs.leidenuniv.nl}

\begin{abstract}
Ensuring privacy of individuals is of paramount importance to social network analysis research.
Previous work assessed anonymity in a network based on the non-uniqueness of a node's ego network.
In this work, we show that this approach does not adequately account for the strong de-anonymizing effect of distant connections. 
We first propose the use of \emph{$d$-$k$-anonymity}, a novel measure that takes knowledge up to distance $d$ of a considered node into account.
Second, we introduce \emph{anonymity-cascade}, which exploits the so-called infectiousness of uniqueness: mere information about being connected to another unique node can make a given node uniquely identifiable. 
These two approaches, together with relevant ``twin node'' processing steps in the underlying graph structure, offer practitioners flexible solutions, tunable in precision and computation time. 
This enables the assessment of anonymity in large-scale networks with up to millions of nodes and edges. 
Experiments on graph models and a wide range of real-world networks show drastic decreases in anonymity when connections at distance $2$ are considered. 
Moreover, extending the knowledge beyond the ego network with just one extra link often already decreases overall anonymity by over 50\%. 
These findings have important implications for privacy-aware sharing of sensitive network data.
\end{abstract}

\begin{document}

\maketitle

    \flushbottom
    \begin{figure}[h]
            \centering
            \scalebox{1}[1]{
            \includegraphics[width=\textwidth]{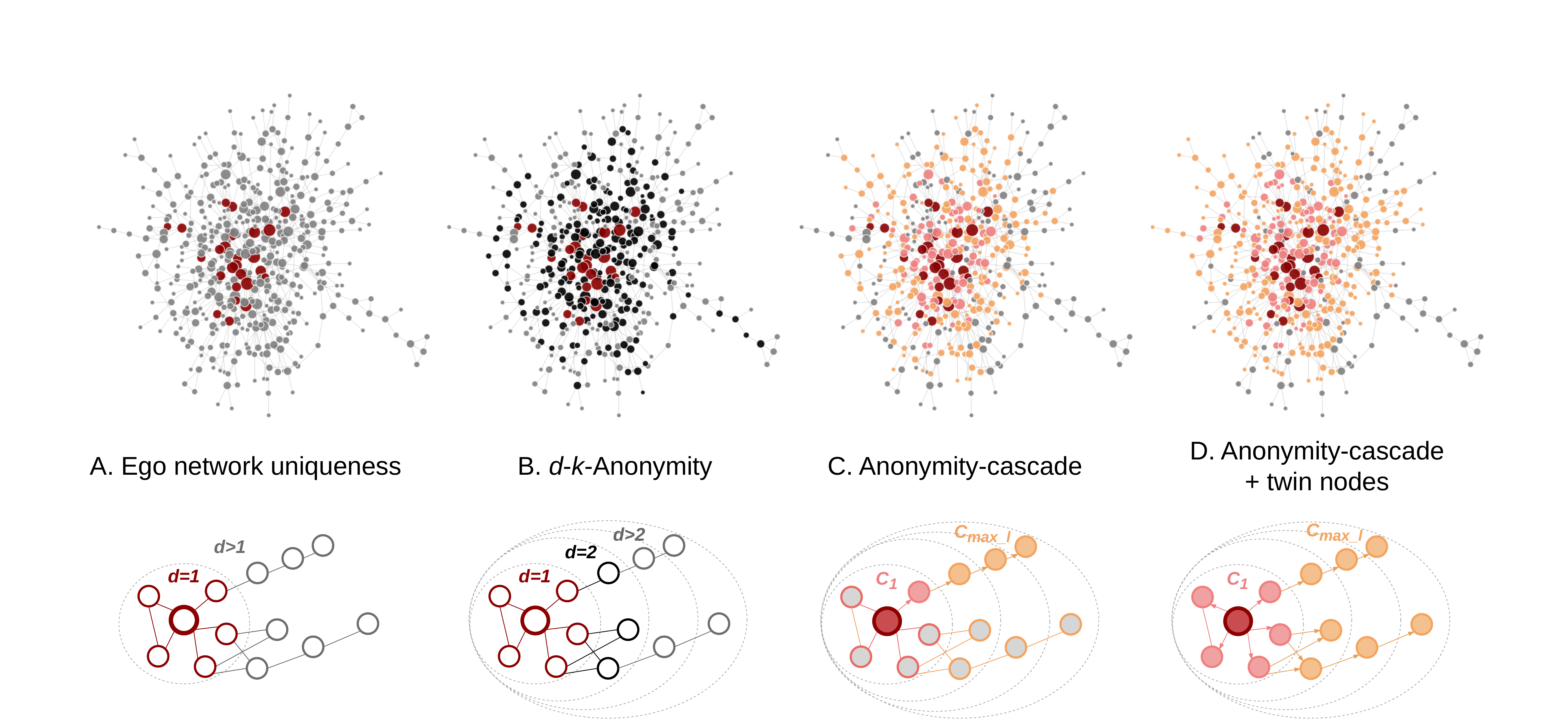}}
            \caption{
            Four approaches for assessing node anonymity: Ego network uniqueness~\cite{romanini2020privacy} (A), followed by the three techniques discussed in this paper: \emph{$d$-$k$-anonymity}~\cite{dejong2023algorithms} (B), \emph{anonymity-cascade} (C) and \emph{anonymity-cascade} with twin nodes (D). \\
            For each approach, the top row shows the uniquely identified nodes in the giant component of the Copnet calls network.~\cite{sapiezynski2019copenhagen} 
            Red nodes are unique using \emph{$1$-$k$-anonymity} (i.e., ego network uniqueness), black nodes with \emph{$2$-$k$-anonymity} (subfigure B only). 
            Pink nodes can be identified using one cascading step ($C_1$), orange nodes with multiple steps ($C_\mathit{max-\ell}$).
            Grey nodes are not uniquely identified using the considered approach.
            The bottom row illustrates an example of a $d$-neighborhood, detailing which knowledge is taken into account by each approach (edge and node outline color).
            Subfigure C and D show the paths traversed by \emph{anonymity-cascade} to identify the pink and orange nodes, given that the red center node is unique for \emph{$1$-$k$-anonymity}.}
           \label{fig:example}
    \end{figure}

\section*{Introduction}\label{sec:intro}

Network science research~\cite{barabasi2016network} is typically about getting a better understanding of the connected structure of a group of people,~\cite{bokanyi2023anatomy} organizations,~\cite{garcia2017uncovering} infrastructural objects~\cite{guimera2004modeling} or other relevant interacting entities.~\cite{barabasi2004network} 
Conducted analyses are often useful for shedding light on
societally relevant problems, such as resilience of technical systems,~\cite{ganin2016operational} predicting systematic financial risk,~\cite{cimini2015systemic} modelling epidemic disease spread~\cite{azizi2020epidemics} or the measurement of socio-economic segregation in a society.~\cite{bojanowski2014measuring, kazmina2023socio}
Crucial for this type of research is the availability of network data representing the interactions that are the object of study.
While pseudonymization is often used to mask the identity of individuals in, for example, a social network dataset, the network structure itself may reveal sensitive information on ``who is who''.
As a result, sharing network data imposes risks on the privacy of the entities represented in it. 
In this paper we set out to discover how we can adequately measure and assess anonymity in complex networks, focusing on methods for discovering how revealing an individual's connections in a network really are. 

Related work on privacy in network centers around two major approaches:~\cite{ji2016graph, li2023private, beigi2020survey} differential privacy~\cite{sala2011sharing, proserpio2012calibrating, wang2013preserving} and $k$-anonymity.~\cite{liu2008towards, romanini2020privacy, zhou2008preserving, hay2007anonymizing, zou2009k, wu2010k}
The first gives randomized answers to user queries such that the privacy of entities are guaranteed, whereas the second enables the sharing of an anonymized version of the network such that there are at least $k$ candidates for each entity. 
Both of these approaches are strongly embedded in the field of Statistical Disclosure Control (SDC), where traditionally, privacy in relational data is studied.~\cite{willenborg2012elements, hundepool2012statistical}
However, network data introduces new challenges since the nodes, by which entities are represented, are not isolated observations.
Unlike tabular data, the anonymity of a node in a network does not solely depend on the node itself, but can be affected by direct and indirect neighbors in the network. 
This comes with substantial methodological and computational challenges related to measuring anonymity in network data. 
\added{A number of different works on anonymity in networks has been published, including various surveys that give a more elaborate overview of this type of work.}~\cite{ji2016graph, li2023private, beigi2020survey}
There exist various works on differential privacy, which aims to provide privacy-preserving answers to queries about the network, possibly to eventually generate a synthetic network data based on anonymized graph properties.~\cite{proserpio2012calibrating, sala2011sharing, wang2013preserving}
\added{Since in this paper we are interested in preventing identity disclosure and ultimately sharing an altered anonymized version of the full network, we have chosen to extend upon the existing line of research}~\cite{liu2008towards, romanini2020privacy, zhou2008preserving, hay2007anonymizing, zou2009k, wu2010k} of $k$-anonymity. 

In this paper, we use the notion of $k$-anonymity and investigate the risk of identity disclosure of nodes based on structural properties of the focal node's surroundings.
Noteworthy is that in the remainder of this work, we use the term ``anonymity'' as a concrete and measurable operationalization of ``privacy''. 
We say that a node is $k$-anonymous if there are $k-1$ equivalent nodes in the network according to a particular measure of equivalence. 
The larger the value of $k$ for a node, the more anonymous the node is.
A network as a whole is said to be $k$-anonymous if all nodes are at least $k$-anonymous.

As we will see when we turn to our experimental results, in some, but definitely not in all networks does the chosen value of $k$\added{, ranging from $k=1$ to 5,} strongly affect overall network anonymity beyond $k=2$.
A value of $k = 2$ corresponds to the situation where a node is anonymous if it is not unique based on the employed definition of equivalence.
In this particular case, on which we largely focus in this paper,  anonymity is effectively equal to non-uniqueness. 
Various equivalence measures have been used in the literature, taking into account the degree,~\cite{liu2008towards} the ego network structure,~\cite{romanini2020privacy, zhou2008preserving} the degree distributions of neighboring nodes~\cite{hay2007anonymizing} or the orbits~\cite{zou2009k, wu2010k} of the node under consideration. 
These measures range from very lenient, accounting for merely the number of connections (degree), to very strict, where nodes are equivalent if they are not distinguishable based on their precise structural position in the network.

All equivalence measures mentioned above correspond to a specific attacker scenario where we assume that someone who tries to de-anonymize entities in the network has a certain type and amount of information.
This introduces a trade-off.
While using a very strict measure would protect against more attacker scenarios, it would at the same time result in fewer $k$-anonymous nodes.
As a result, when one aims to anonymize the network, e.g., by means of perturbation,~\cite{zhou2008preserving} more changes may be required to ensure that all nodes and therewith the network are $k$-anonymous.
This might have a major impact on the similarity to the original network and hence the so-called utility of the resulting anonymized network.
When choosing a measure it is therefore important to account for a realistic amount of attacker information and therewith protect against realistic attacker scenarios.
At the same time, the measure should be computable in a reasonable amount of time for nowadays common network sizes of potentially millions of nodes and edges. 

In this paper, we aim to contribute to existing literature on this topic in three ways.
First, we show the effect on anonymity when one has knowledge beyond the ego network.
Empirical results employing the parameterized measure of  \emph{$d$-$k$-anonymity},~\cite{dejong2023algorithms}\added{ for which parameter $d$ denotes the distance from the considered node,} indicate that the largest decrease in anonymity occurs when considering $2$-neighborhoods (shown in black in Fig.~\ref{fig:example}B) rather than just the ego networks (1-neighborhoods, shown in red in Fig.~\ref{fig:example}B.). 
This holds for both well-known graph models and a wide range of real-world networks.  
Second, we aim to better understand the so-called infectiousness of uniqueness
in networks by introducing \emph{anonymity-cascade} (Fig.~\ref{fig:example}C).
This approach extends the aforementioned approach of \emph{$d$-$k$-anonymity} by means of a cascading step that finds all nodes that can be uniquely identified if an attacker knows that a particular node is connected to a specific unique node, as illustrated by the pink nodes in Fig.~\ref{fig:example}C.
The newly identified nodes can be reused iteratively,  which can result in a cascading effect as illustrated by the orange nodes in Fig.~\ref{fig:example}C. 
Our results on a diverse set of real-world networks demonstrate that even knowledge of one extra link, i.e., conducting one cascading ``step'', frequently reduces overall anonymity by over 50\%. 
Third, we show how regularities in the underlying graph structure, specifically twin nodes,~\cite{gonzalez2019removing} which frequently occur in real-world networks, can be exploited to obtain additional information on certain otherwise indistinguishable entities. 
This is illustrated in Fig.~\ref{fig:example}D.

The remainder of this paper is structured as follows.
In~\hyperref[sec:results]{Results}, we discuss findings resulting from each of the three newly proposed approaches illustrated in Fig.~\ref{fig:example}B-D, starting with 
the parameterized anonymity measure and the cascading algorithm. 
For both, we present results on graph models and real-world networks, before ending with a third and final subsection on the de-anonymizing effect of aforementioned twin nodes. 
We conclude the paper by summarizing the most important results together with possible directions for future work in~\hyperref[sec:con]{Discussion}.
Details about the three approaches, the overall experimental setup, code ensuring reproducibility, as well as relevant theorems and proofs, can be found in~\hyperref[sec:methods]{Methods}.

\section*{Results}\label{sec:results}
In this section, we discuss the three approaches to assess node anonymity in a network, each illustrated in Fig.~\ref{fig:example}.
First, in~\hyperref[sec:meas1]{Beyond the ego network}, we look at the de-anonymizing effect that knowledge about $d$-neighborhoods can have when $d\geq 2$.
Second, in~\hyperref[sec:casc]{Anonymity-cascade}, we extend \emph{$d$-$k$-anonymity} with a cascading step to capture the ``infectiousness of uniqueness''.
For both approaches, we discuss results on both graph models and real-world networks.
Third, in~\hyperref[sec:twin]{Twin nodes}, we look at an approach to leverage regularities in the underlying graph structure and identify even more nodes than with the two aforementioned approaches.
\subsection*{Beyond the ego network}\label{sec:meas1}

In this section, we investigate the effect of knowledge beyond the ego network on anonymity of nodes by using the notion of \emph{$d$-$k$-anonymity}.~\cite{dejong2023algorithms}
We say that two nodes are $d$-equivalent if they are indistinguishable with perfect knowledge about their $d$-neighborhood and position in this neighborhood (See Definition~\ref{def:equivalent} in~\hyperref[sec:methods]{Methods}).
Here, the $d$-neighborhood consists of the node itself, all nodes that can be reached by traversing at most $d$ edges, and all edges between these nodes.
When $d=1$, this corresponds to the ego network of the node.
This is also illustrated in the bottom of Fig.~\ref{fig:example}A-B.
More precisely, for a pair of $d$-equivalent nodes the $d$-neighborhoods are isomorphic (See Definition~\ref{def:iso} in~\hyperref[sec:methods]{Methods}) and their respective position in the $d$-neighborhood is the same.

If, for a specific node, there are $k-1$ nodes to which it is $d$-equivalent, the node is in an equivalence class of size $k$ and we say that the node is $d$-$k$-anonymous. 
If the node is $d$-$1$-anonymous, we also call it unique.
We summarize the uniqueness of a network 
as the fraction of unique nodes.
Thus, a high network uniqueness implies low anonymity, and low uniqueness implies high anonymity. 
The results for \emph{$d$-$k$-anonymity} are computed by the algorithms described in previous work~\cite{dejong2023algorithms} that builds upon a state-of-the-art isomorphism computation tool~\cite{nauty} (see~\hyperref[sec:methods]{Methods} for details).
In the following sections, we discuss results of using the measure of \emph{$d$-$k$-anonymity} on both graph models and a wide range of real-world networks.

    \begin{figure}[!t]
           \centering
           \includegraphics[width=0.9\textwidth]{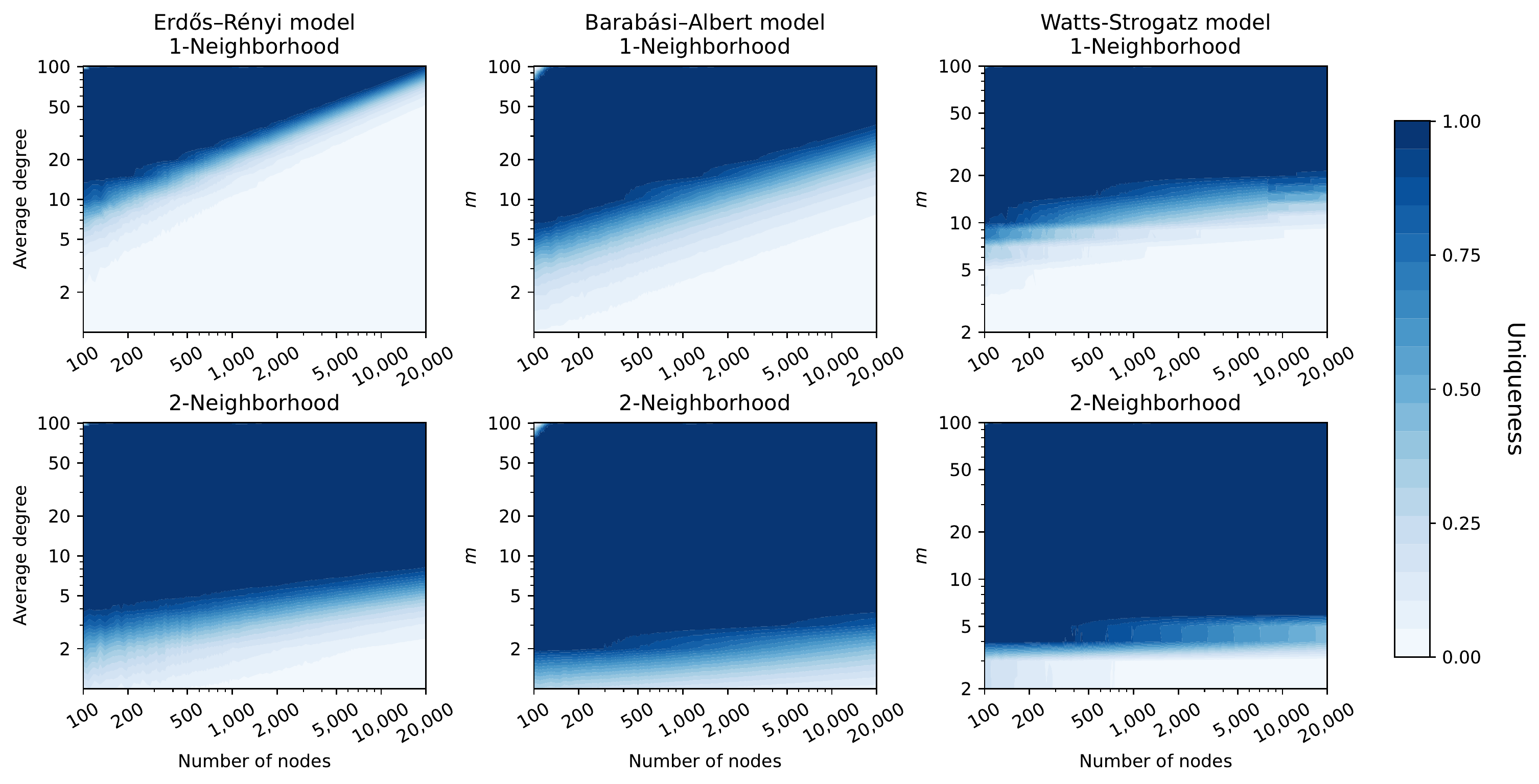}
           \caption{Uniqueness maps using $d$-$k$-anonymity. Maps show network uniqueness, indicated by color, when using information of the 1-neighborhood (top row) and 2-neighborhood (bottom row). Results are shown for the Erdős–Rényi (left), Barabási–Albert (middle) and Watts–Strogatz (right) model with different sizes (horizontal axis) and average degree or\added{ $m$, an } equivalent thereof (vertical axis).
           }
           \label{fig:modelsdk}
    \end{figure}

\subsubsection*{Beyond the ego network in graph models}\label{sub:mod1}

In this section we investigate the uniqueness of networks (i.e., the fraction of unique nodes) when a  possible attacker has perfect knowledge about the ego network or $2$-neighborhood of a node.
We use three common graph models that each generate networks that reflect a different property that is often observed in real-world networks.
For each model, we vary in size and density.
The first graph model, the Erdős–Rényi (ER) model,~\cite{erdos1960evolution} generates edges completely at random.
Second, the Barabási-Albert (BA) model~\cite{barabasi1999emergence} generates edges by means of the preferential attachment mechanism, which results in the skewed degree distribution that is frequently observed in real-world networks.
Third, the Watts-Strogatz (WS) model~\cite{watts1998collective} additionally captures the small world property.
More details about, for example, the used parameters, 
can be found in~\hyperref[sec:methods]{Methods}.

In Fig.~\ref{fig:modelsdk} the results on graph models are shown.
The figures correspond to the uniqueness maps used by Romanini \emph{et al.}~\cite{romanini2020privacy} where the horizontal axis shows the number of nodes, the vertical axis denotes the average degree or $m$, \added{which equals the number connections made per node for the BA model, or the number initial connections for each node for the WS model}\deleted{ indicating the density}. 
The color indicates uniqueness of the graph ranging from 0.0 (white, no unique nodes) to 1.0 (dark blue, all nodes unique).
Each result is averaged over ten generated graphs.

The results using knowledge of the ego network (top) correspond to the results by Romanini \emph{et al.}~\cite{romanini2020privacy} and show a clear connection between the average degree and the fraction of unique nodes.
When the number of nodes grows and the average degree is constant, this fraction tends to decrease, meaning that nodes are overall more anonymous.
Moreover, these figures show a very clear turning point: \added{for the ranges shown,} below the white line almost no nodes are unique while above this line almost all nodes are unique based on their ego network.
\added{Results in~\hyperref[sec:supmat]{Supplementary information} show that this also holds for higher densities, except when the graph is (near)-complete; for this in real-world networks unrealistic setting, all nodes become non-unique.}

However, when we look at the results on the 2-neighborhood computed using \emph{$d$-$k$-anonymity} (bottom row of Fig.~\ref{fig:modelsdk}), we see that the uniqueness increases significantly for all models.
After an average degree of five almost all nodes are unique, and the uniqueness does not strongly decrease as the network size grows.
This shows a large contrast to the results of the 1-neighborhood.
Interestingly for $d=3$ and higher, no large changes occur, which implies that the largest de-anonymizing effect occurs for $d = 2$ (see results up to $d = 5$ in the~\hyperref[sec:supmat]{Supplementary information}).

\begin{figure}[t!]
    \centering
    \includegraphics[width=0.8\linewidth]{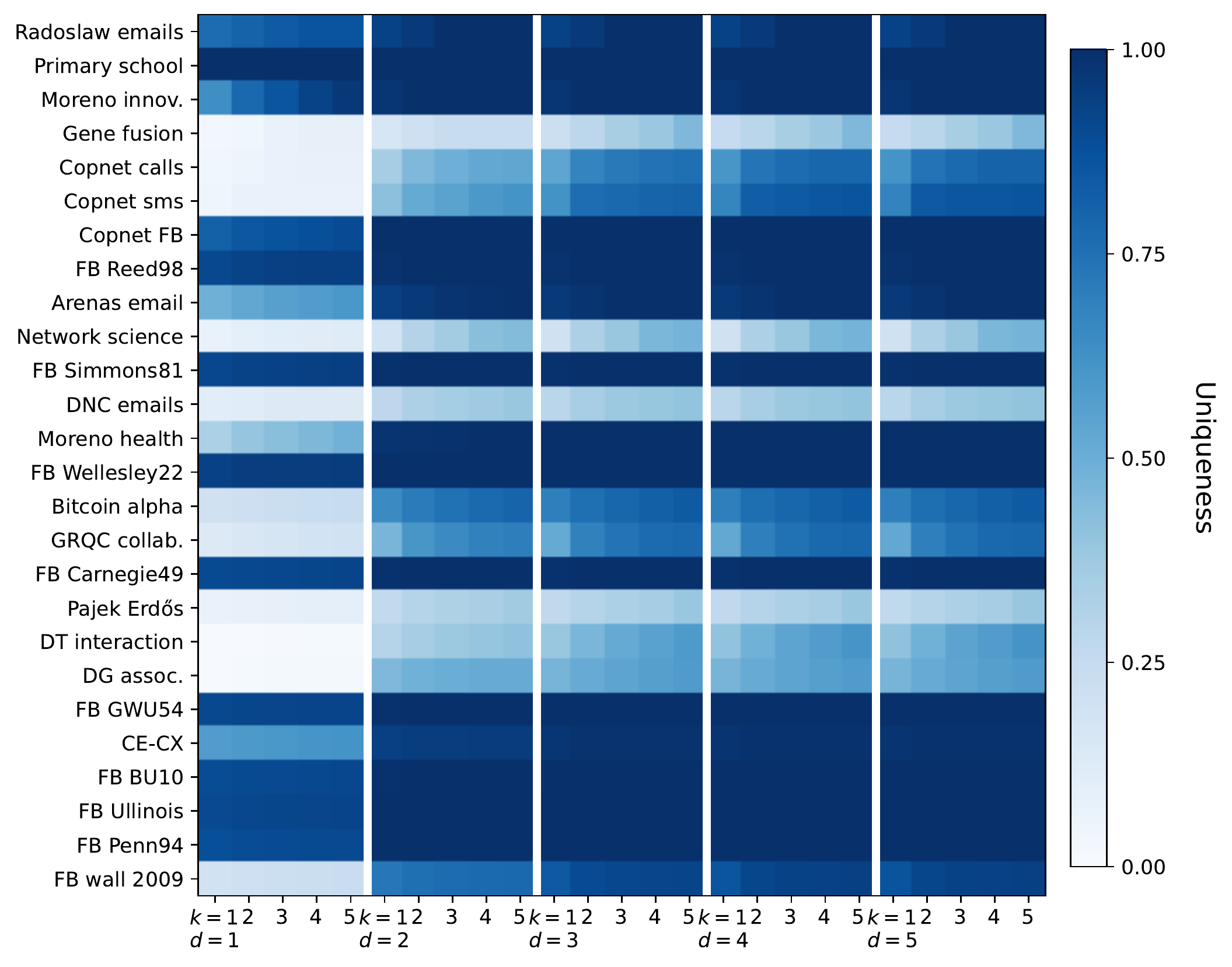}
    \caption{$d$-$k$-Anonymity in real-world networks. Results are shown for for the 29 real-world networks in Table~\ref{tab:data} for which \emph{$d$-$k$-anonymity} with $d=5$ could be computed within three hours. Each cell denotes the fraction (ranging from 0.0 (white) to 1.0 (dark blue)) of nodes that are $\leq k$-anonymous when accounting for knowledge of the $d$-neighborhood. 
    }
    \label{fig:dk-real}
\end{figure}

\begin{table}[t!]
\scriptsize
\centering
\begin{tabular}{@{}rrrrrrrrrl@{}}

\toprule
    & \multicolumn{1}{l}{} & \multicolumn{1}{l}{} & \multicolumn{1}{l}{} & \multicolumn{1}{l}{} & \multicolumn{2}{c}{\begin{tabular}[c]{@{}c@{}}Runtime \\ \emph{$d$-$k$-anonymity}\end{tabular}} & \multicolumn{2}{c}{\begin{tabular}[c]{@{}c@{}}Runtime\\ \emph{Anonymity-cascade}\end{tabular}} & \multicolumn{1}{l}{}  \\
\multicolumn{1}{r}{Network}    & \multicolumn{1}{r}{Nodes} & \multicolumn{1}{r}{Edges} & \multicolumn{1}{r}{Fraction twins} & \multicolumn{1}{r}{$max-\ell$} & \multicolumn{1}{c}{d=1} & \multicolumn{1}{c}{d=2} & \multicolumn{1}{c}{$C_1$} & \multicolumn{1}{c}{$C_\mathit{max-\ell}$} & \multicolumn{1}{l}{Type}                      \\ \midrule
Radoslaw emails~\cite{kunegis2013konect}          & 167        & 3,250     & 0.072        & 3        & 0.02    s & +  < 0.01 s    & + < 0.01 s & + < 0.01  s    & Communication         \\
Primary school~\cite{sociopatterns}               & 236        & 5,899     & 0.000        & 1        & 0.02    s & +  < 0.01 s    & + < 0.01 s & + < 0.01  s    & Human contact         \\
Moreno innov.~\cite{kunegis2013konect}            & 241        & 923       & 0.025        & 3        & < 0.01  s & +  < 0.01 s    & + < 0.01 s & + < 0.01  s    & Communication         \\
Gene fusion~\cite{kunegis2013konect}              & 291        & 279       & 0.753        & 6        & < 0.01  s & +  < 0.01 s    & + < 0.01 s & + < 0.01  s    & Biological            \\
Copnet calls~\cite{sapiezynski2019copenhagen}     & 536        & 621       & 0.287        & 10       & < 0.01  s & +  0.01   s    & + < 0.01 s & + < 0.01  s    & Communication         \\
Copnet sms~\cite{sapiezynski2019copenhagen}       & 568        & 697       & 0.285        & 7        & < 0.01  s & +  0.01   s    & + < 0.01 s & + < 0.01  s    & Communication         \\
Copnet FB~\cite{sapiezynski2019copenhagen}        & 800        & 6,418     & 0.005        & 4        & 0.02    s & +  0.01   s    & + < 0.01 s & + < 0.01  s    & Online social         \\
FB Reed98~\cite{networksrepository}               & 962        & 18,812    & 0.012        & 3        & 0.05    s & +  0.01   s    & + 0.01   s & + 0.02    s    & Online social         \\
Arenas email~\cite{kunegis2013konect}             & 1,133      & 5,451     & 0.042        & 5        & 0.02    s & +  0.01   s    & + < 0.01 s & + < 0.01  s    & Communication         \\
Network science~\cite{kunegis2013konect}          & 1,461      & 2,742     & 0.755        & 6        & 0.01    s & +  0.01   s    & + < 0.01 s & + < 0.01  s    & Co-autorship          \\
FB Simmons81~\cite{networksrepository}            & 1,518      & 32,988    & 0.011        & 3        & 0.12    s & +  0.02   s    & + 0.01   s & + 0.02    s    & Online social         \\
DNC emails~\cite{kunegis2013konect}               & 1,893      & 4,385     & 0.706        & 4        & 0.02    s & +  0.38   s    & + < 0.01 s & + < 0.01  s    & Online social         \\
Moreno health~\cite{kunegis2013konect}            & 2,539      & 10,455    & 0.003        & 5        & 0.03    s & +  0.04   s    & + < 0.01 s & + < 0.01  s    & Human social          \\
FB Wellesley22~\cite{networksrepository}          & 2,970      & 94,899    & 0.000        & 4        & 0.4     s & +  0.07   s    & + 0.04   s & + 0.08    s    & Online social         \\
Bitcoin alpha~\cite{networksrepository}           & 3,783      & 14,124    & 0.306        & 5        & 0.04    s & +  0.36   s    & + < 0.01 s & + 0.01    s    & Online social (trust) \\
GRQC collab.~\cite{snapnets}                      & 5,242      & 14,484    & 0.455        & 8        & 0.03    s & +  0.04   s    & + < 0.01 s & + < 0.01  s    & Co-autorship          \\
FB Carnegie49~\cite{networksrepository}           & 6,637      & 249,967   & 0.008        & 4        & 1.37    s & +  0.44   s    & + 0.09   s & + 0.18    s    & Online social         \\
Pajek Erdős~\cite{kunegis2013konect}              & 6,927      & 11,850    & 0.737        & 6        & 0.02    s & +  0.19   s    & + 0.01   s & + 0.02    s    & Co-autorship          \\
DT interaction~\cite{biosnapnets}                 & 7,341      & 15,138    & 0.572        & 12       & 0.07    s & +  2.88   s    & + < 0.01 s & + < 0.01  s    & Biological            \\
DG assoc.~\cite{biosnapnets}                      & 7,813      & 21,357    & 0.531        & 8        & 0.21    s & +  5.13   s    & + < 0.01 s & + < 0.01  s    & Biological            \\
FB GWU54~\cite{networksrepository}                & 12,193     & 469,528   & 0.006        & 4        & 2.64    s & +  0.92   s    & + 0.18   s & + 0.36    s    & Online social         \\
Anybeat~\cite{networksrepository}                 & 12,645     & 49,132    & 0.500        & 5        & 0.41    s & +  1.34   h    & + 0.02   s & + 0.04    s    & Online social         \\
CE-CX~\cite{networksrepository}                   & 15,229     & 245,952   & 0.021        & 6        & 1.04    s & +  1.56   s    & + 0.09   s & + 0.19    s    & Biological            \\
Astro Physics~\cite{networksrepository}           & 18,771     & 198,050   & 0.305        & 6        & 0.72    s & +  1.91   s    & + 0.05   s & + 0.11    s    & Co-autorship          \\
FB BU10~\cite{networksrepository}                 & 19,700     & 637,528   & 0.006        & 4        & 3.17    s & +  1.44   s    & + 0.25   s & + 0.50    s    & Online social         \\
FB Uillinois~\cite{networksrepository}            & 30,664     & 1,048,574 & 0.002        & 4        & 5.96    s & +  2.66   s    & + 0.39   s & + 0.78    s    & Online social         \\
Enron email~\cite{snapnets}                       & 36,692     & 183,831   & 0.528        & 6        & 0.90    s & +  1.44   m    & + 0.06   s & + 0.12    s    & Communication         \\
FB Penn94~\cite{networksrepository}               & 41,536     & 1,362,220 & 0.002        & 4        & 8.95    s & +  5.39   s    & + 0.50   s & + 1.00    s    & Online social         \\
FB wall 2009~\cite{kunegis2013konect}             & 46,952     & 183,412   & 0.133        & 8        & 0.54    s & +  1.80   s    & + 0.05   s & + 0.13    s    & Communication         \\
Brightkite~\cite{networksrepository}              & 58,228     & 214,078   & 0.258        & 8        & 0.78    s & +  18.41  s    & + 0.06   s & + 0.14    s    & Online social         \\
The marker cafe~\cite{dataforgoodlab}             & 69,413     & 1,644,843 & 0.200        & 5        & 37.96   s & +  10.46  m    & + 0.67   s & + 1.35    s    & Human contact         \\
Slashdot zoo~\cite{kunegis2013konect}             & 79,116     & 467,731   & 0.274        & 7        & 3.36    s & +  2.62   m    & + 0.15   s & + 0.34    s    & Online social         \\
Twitter~\cite{kunegis2013konect}                  & 465,017    & 833,540   & 0.801        & 5        & 59.53   s & +  6.04   h    & + 0.29   s & + 0.62    s    & Online social         \\
DBLP~\cite{kunegis2013konect}                     & 1,824,701  & 8,344,615 & 0.402        & 10       & 51.37   s & +  11.58  m    & + 28.1   s & + 57.18   s    & Co-autorship          \\
Flixster~\cite{kunegis2013konect}                 & 2,523,386  & 7,918,801 & 0.631        & 8        & 3.10    m & +  8.57   h    & + 18.58  s & + 37.81   s    & Online social         \\
Youtube~\cite{kunegis2013konect}                  & 3,223,585  & 9,375,374 & 0.336        & 15       & 13.24   m & +  > 1    week & + 30.68  s & + 1.04    m     & Online social        \\
\bottomrule
\end{tabular}
\caption{Overview of the real-world networks used in the experiments. For each network, we list the number of nodes, edges, fraction of twin nodes, the highest attained cascading level ($max-\ell$) and runtimes of the experiments performed containing the total runtime ($d=1)$ and runtime additional to computing $d=1$, indicated by ``+'' for $d=2$, $C_1$ and $C_\mathit{max-\ell}$.}
\label{tab:data}
\end{table}

\subsubsection*{Beyond the ego network in real-world network data}\label{sub:emp1}
For the next set of experiments, we used a wide range of real-world networks varying in size, density and category.
All networks are publicly available and can be found in their corresponding repositories cited in Table~\ref{tab:data}, which in addition to various experimental results on runtime (further addressed in~\hyperref[sec:casc]{Anonymity-cascade}) summarizes for each network elementary characteristics such as the number of nodes and edges and the type of network data, covering, e.g., online social networks, co-authorship and biological networks.

In Fig.~\ref{fig:dk-real}, results are shown for a range of real-world networks for which \emph{$d$-$k$-anonymity} could be computed within three hours up to and including $d=5$.~\cite{dejong2023algorithms}
For distance $d=1$ to $d=5$ (the five separated columns), we show for $k=1$  to $k=5$ by means of color intensity which fraction of the nodes is unique. 
By reporting on different values for $k$, we aim to take into account that in some cases not being unique does not ensure sufficient privacy, and larger values for $k$ are commonly used.

For most networks, similar to previously discussed results on graph models, we observe the largest increase in uniqueness when moving from $d=1$ to $d=2$.
On average, the absolute increase in uniqueness 
equals 0.24 with the highest increase being 0.65 for ``Moreno health''.
For 12 out of 26 networks, this additional knowledge more than doubles the number of unique nodes.
After $d=2$ an average increase of 0.04 is observed when increasing to $d=5$, with the largest value of 0.27 for the ``Copnet calls'' network.
This shows that for most networks, moving from knowledge about the ego network to knowledge about the 2-neighborhood has the largest effect on anonymity.
In the figure, we can overall distinguish between three different cases: 
1) a high uniqueness at $d=1$, or if there is a low uniqueness at $d=1$, there is either
2) a high uniqueness at $d=2$, or
3) a low uniqueness at all distances.
\added{In \hyperref[sec:supmat]{Supplementary information}, we include results aiming to correlate the uniqueness found to various graph properties.
For the networks presented in Table~\ref{tab:data}, we find that networks with a larger diameter or average path length tend to have a lower uniqueness. 
Networks with high degrees or density tend to have a higher uniqueness. This is also shown by the work of Romanini \emph{et al.}~\cite{romanini2020privacy} and is similar to earlier results obtained for the graph models.}

Additionally, we compare different values for $k$, where we measure the fraction of nodes that are at most $k$-anonymous, and hence the fraction of nodes for which there are at most $k$ candidates for an attacker with knowledge of their $d$-neighborhood.
Increasing the value for $k$ to 5 results in an average increase of 0.05 to 0.08 with the largest increase equal to 0.33 for the ``Moreno innovation'' communication network.
The results show that in many cases larger values beyond $k=2$ \added{up to $k=5$} do not result in a large decrease in anonymity.
Thus, we can learn a lot by only distinguishing between unique and non-unique nodes to measure anonymity.
With this in mind, and in the interest of readability of further results, we choose to report on uniqueness ($k = 1$) and \emph{$d$-$k$-anonymity} with $d=1$ and $d=2$ in the remainder of this paper.
Overall, we conclude that accounting for knowledge beyond the ego network in both graph models and real-world networks shows a significant decrease in node anonymity.
For completeness, a specific figure showing the uniqueness using different values for $d$ can be found in~\hyperref[sec:supmat]{Supplementary information}.
\subsection*{Anonymity-cascade}\label{sec:casc}

The measure of \emph{$d$-$k$-anonymity} employed above, while more informative than ego network uniqeueness, has two noteworthy disadvantages.
First, due to isomorphism computations of possibly large neighborhoods, for larger values of $d$, this approach is computationally expensive~\cite{dejong2023algorithms} (see also the runtimes for $d=2$ in the seventh column of Table~\ref{tab:data}).
Second, knowledge of the $2$-neighborhoods can be an unrealistic attacker scenario; in particular if the 2-neighborhood has a complex structure.
However, it is not unreasonable to assume that an attacker obtains some information beyond the ego network, especially if the 1-neighborhood is small or the 2-neighborhood sparse.
If that knowledge includes that the node is connected to a unique node, which can be concluded using knowledge of the 1-neighborhood, this may strongly decrease the number of candidates. 
In some cases this can be sufficient to uniquely identify a node.

We propose to explicitly detect this by introducing \emph{anonymity-cascade} ($C_{\ell}$), an algorithm that extends \emph{$d$-$k$-anonymity} and accounts for the so-called ``infectiousness of uniqueness''.
We assume that an attacker has knowledge about the 1-neighborhoods of two nodes, of which one is unique, and that there is a connection between them.
The \emph{anonymity-cascade} algorithm starts by finding all nodes that can be uniquely identified by knowing that this node 1) is connected to a unique node $u$ using \emph{$1$-$k$-anonymity} and 2) is unique in the set of neighbors of node $u$. 
We refer to this as the first level of the cascading algorithm, denoted $C_1$.
The nodes identified with $C_1$ can then be used to continue this cascading effect among further levels in a Breadth-First Search manner for $\ell$ steps or ``levels'', as illustrated in Fig.~\ref{fig:example}C. 
In this figure, the process starts at the unique red node.
From this node, the pink node can be uniquely identified ($C_1$) and from there the nodes can be used to identify the orange nodes ($C_{\geq2}$).
This process can be repeated until no more unique nodes are found, the then attained level is called $\mathit{max-\ell}$.
We refer to this as cascading final, denoted $C_{max-\ell}$.
While going beyond the first level mimics a less realistic attacker scenario, this approach gives insights into how far the cascading effect can continue (shown in the fifth column of Table~\ref{tab:data}), and the effect this can potentially have.

    \begin{figure}[t!]
            \centering
            \centering
            \includegraphics[width=0.99\textwidth]{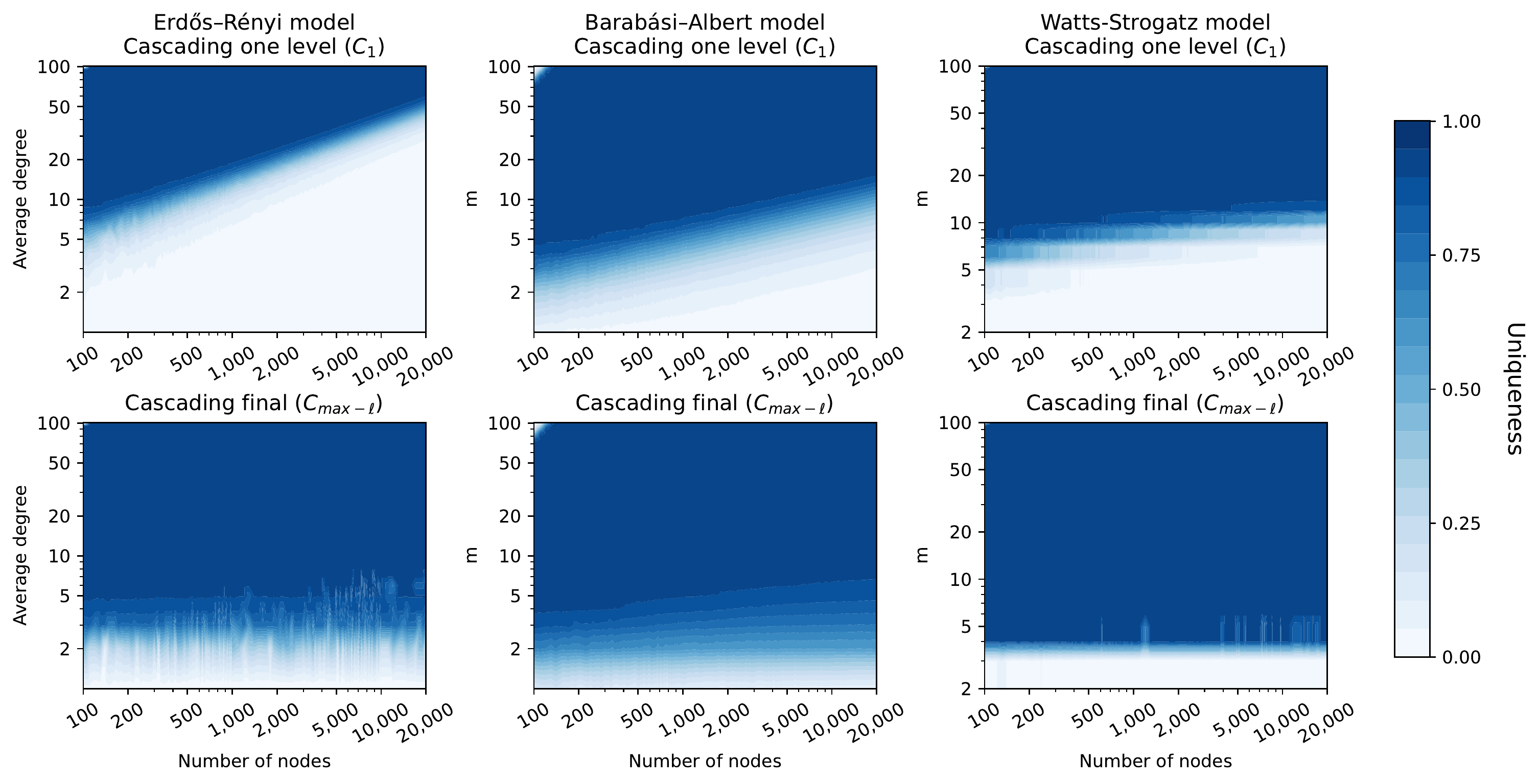}
            \caption{Uniqueness maps using \emph{anonymity-cascade}.
            Maps show network uniqueness, indicated by color, when using one level of cascading (top row) and up to the final level of cascading (bottom row).
            Results are shown for the Erdős–Rényi (left), Barabási–Albert (middle) and Watts–Strogatz (right) model with different sizes (horizontal axis) and average degree or\added{ $m$, an} equivalent thereof (vertical axis).
            }
           \label{fig:modelsc}
    \end{figure}

For one level of cascading ($C_1$), the measure is more strict than \emph{$1$-$k$-anonymity} and at most as strict as \emph{$2$-$k$-anonymity}, which is explained in Theorem~\ref{thm:casc} and Proof~\ref{proof:casc} in~\hyperref[sec:methods]{Methods}.
Additionally, \emph{anonymity-cascade} is less expensive to compute than \emph{$2$-$k$-anonymity} and allows us to assess anonymity in larger networks with millions of nodes, as shown in Table~\ref{tab:data}.
A more detailed description of \emph{anonymity-cascade} can be found in~\hyperref[sec:methods]{Methods}.

\subsubsection*{Anonymity-cascade in graph models}\label{sub:mod2}
In Fig.~\ref{fig:modelsc}, the results of \emph{anonymity-cascade} on graph models can be found, similar to the uniqueness maps in Fig.~\ref{fig:modelsdk}.
The top row shows the results of $C_1$ (one level of cascading), which resemble the results of \emph{$d$-$k$-anonymity} with $d=1$. Apparently, for these graph models, having knowledge about one additional link does not strongly affect the overall anonymity.
When the algorithm continues up to its final level ($C_\mathit{max-\ell}$), as shown in the bottom row, the results change significantly and many more nodes are unique.

The maximal depth reached by \emph{anonymity-cascade} ($max-\ell$) can be very high; averages of 38, 14 and 15 are observed for ER, BA and WS respectively. Especially for sparse graphs with over 15,000 nodes high values are observed (results can be found in~\hyperref[sec:supmat]{Supplementary information}).
Figure~\ref{fig:modelsc} also shows that for $C_\mathit{max-\ell}$ the fraction of unique nodes is more stable when the graph size increases compared to \emph{$2$-$k$-anonymity} in Fig.~\ref{fig:modelsdk} for the ER and WS models.
This seems to indicate that in graph models, which are more random than real-world networks, cascading has a small but local effect that can continue for many levels, especially in large graphs.

    \begin{figure}[b!]
        \centering
        \includegraphics[width=\linewidth]{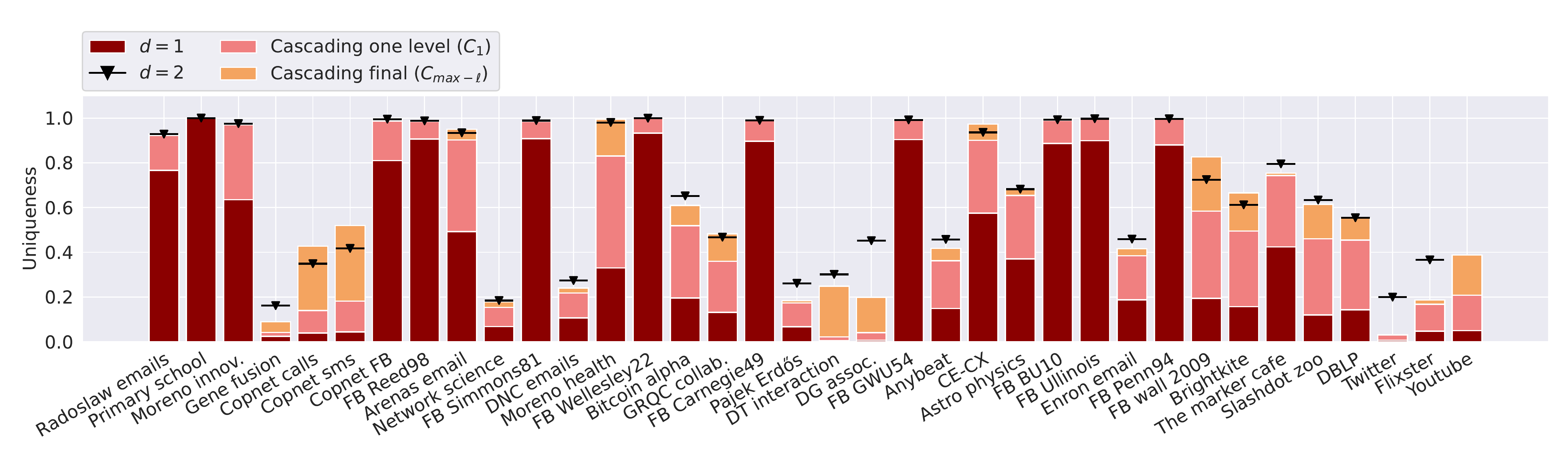}
        \caption{Uniqueness in real-world networks. Fraction of unique nodes (vertical axis) on different datasets (horizontal axis) when accounting for different levels of information: 1-neighborhood (red), 2-neighborhood (black line with triangle), cascading one level (pink) and the cascading final (yellow).}
        \label{fig:cascading}
    \end{figure}

\subsubsection*{Anonymity-cascade in real-world network data}\label{sub:emp2}
The results in Fig.~\ref{fig:cascading} 
show the fraction of unique nodes when using \emph{anonymity-cascade} ($C_1$ and $C_\mathit{max-\ell}$)
compared to \emph{$d$-$k$-anonymity} with $d=1$ and $d=2$.
The results show that when having additional knowledge of one extra link, the fraction of uniquely identifiable nodes doubles for 19 
out of 36 networks.
The largest \emph{anonymity-cascade} increase equals 0.5 for ``Moreno health'': while at $d=1$ a fraction of 0.33 is uniquely identified, one level of cascading results in a uniqueness of 0.83.

The results additionally show that the fractions obtained by $C_1$ are often close to the fraction identified by \emph{$2$-$k$-anonymity}.
On average the difference equals just 0.09.
However, for some networks this difference is larger: for 13 networks this fraction is still larger than 0.1.
Overall, while the runtimes in Table~\ref{tab:data} showed that computing \emph{$2$-$k$-anonymity} can be computationally expensive, especially in large networks, using \emph{anonymity-cascade} provides an adequate estimate of \emph{$2$-$k$-anonymity} in a reasonable amount of time, even for networks with millions of nodes and edges.

When continuing the cascading effect until the the final level ($C_\mathit{max-\ell}$), this results in an additional increase in uniqueness for many networks.
While this level of knowledge is less realistic as an attacker scenario, this gives insights into how far one could exploit this cascading effect.
For five 
of the networks, this results in a uniqueness increase larger than a factor two.
The cascading approach can hence be effective, even at higher levels. 
\added{As shown in Table~\ref{tab:data}, the value for $max-\ell$ achieved can differ per network. In ~\hyperref[sec:supmat]{Supplementary information}, we include results aiming to explain the difference. Large average path lengths and diameter seem to result in longer possible cascading paths for these networks, which intuitively makes sense given that the value is based on paths. Higher degrees and densities are overall likely to result in shorter paths.}
However, contrary to what the results for graph models show, for most real-world networks the largest increase in uniqueness happens at $C_1$ (see~\hyperref[sec:supmat]{Supplementary information} for a figure detailing the drastically decreasing effect per subsequent level).
\subsection*{Twin nodes}\label{sec:twin}

The previous two approaches can be extended by means of  a ``twin node''~\cite{gonzalez2019removing} processing step. 
This concerns the identification of sets of nodes that all share the exact same neighbors \added{which we refer to as twin nodes}. 
Focusing on the example of two pairs of such nodes in Fig.~\ref{fig:twin}B-C, we distinguish two cases: either the nodes are connected to the same nodes (open twin nodes, Fig.~\ref{fig:twin}B), or they additionally have a connection between each other (closed twin nodes, Fig.~\ref{fig:twin}C). 
This differs from the case illustrated in Figure~\ref{fig:twin}A.
If a full network is shared in a pseudonymized format, such as the networks listed in Table~\ref{tab:data}, then sets of twin nodes can be derived from the network itself, without requiring any additional external information.
It turns out that twin nodes can occur frequently in real-world networks: in Table~\ref{tab:data}, fractions up to 0.801 are observed. 

In Proof~\ref{proof:twin}, we show that twin nodes are indistinguishable based on any structural property, implying that any structure-based measure for equivalence can not distinguish between these nodes.
As a result\deleted{s}, if a node has at least one twin, it can not be unique. 
At the same time, if in an attempt to identify entities\deleted{, given a certain amount of knowledge,} all candidates for a node are twin, then this gives the same information as when the node is uniquely identified in the network. 

The reason is that for these nodes we know both their exact structural position in the network and which nodes they are connected to, which is the same information if a node is uniquely identified. 
This notion relates to group disclosure in SDC:~\cite{aggarwal2008general} even if there are \replaced{multiple candidates}{are groups of $k$ candidates} for an entity, it might still be possible to derive sensitive information if all candidates have the same attribute\replaced{. I}{, i}n our situation,\replaced{ their structural position in the network and the connections of the node}{ if they are twins}. 
\added{In SDC, the problem is overcome by introducing the notion of $\ell$-diversity~\cite{machanavajjhala2007diversity} which extends $k$-anonymity with the requirement that for each of the $k$ candidates, there should be at least $\ell$ different sensitive values.}

We process twin nodes in our approaches as follows.
First, for \emph{$d$-$k$-anonymity}, we say that a node is $twin$-unique if either the node is unique, or all candidates for the node are twins of each other.
Second, in \emph{anonymity-cascade}, we start with all nodes that are twin-unique using \emph{$1$-$k$-anonymity}.
If in the cascading step all candidates are twins, they are also twin-unique, and we continue the cascading process from each of the nodes (see also the example in Fig.~\ref{fig:example}D).

    \begin{figure}[!b]
        \centering
        \includegraphics[width=1.0\linewidth]{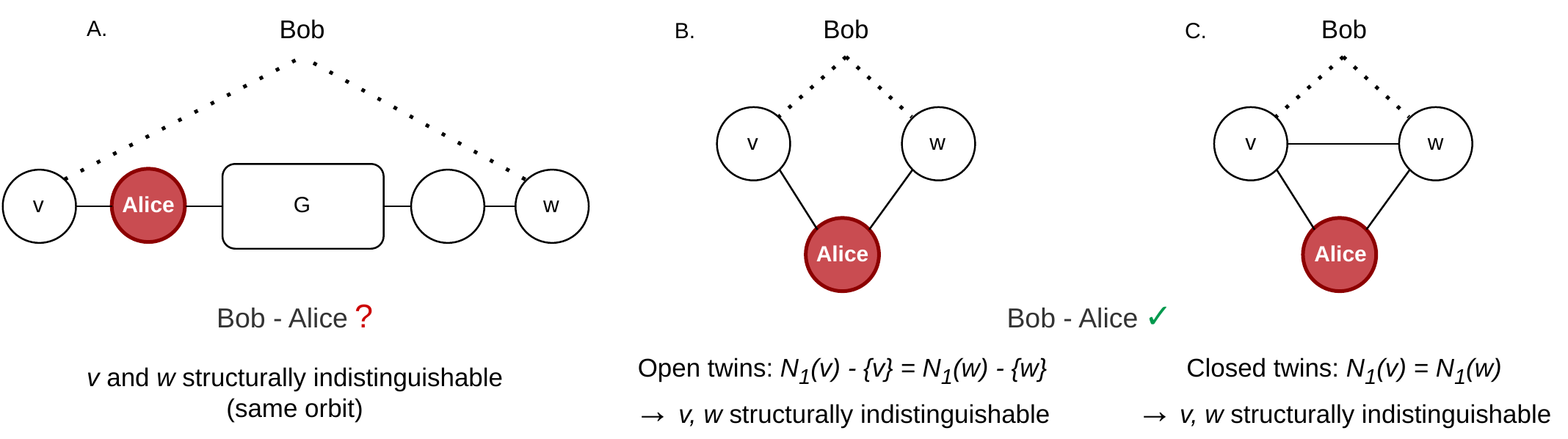}
        \caption{Illustration of two types of structurally indistinguishable nodes. (A) Two candidate nodes for Bob (indicated by dotted lines) are structurally indistinguishable, but do not share their connections. 
        Given that the position of Alice in the network is known (e.g. using \emph{$1$-$k$-anonymity},
        there is no certainty about the connection between Bob and Alice.
       (B, C) All candidates for Bob are twin-unique. The nodes are structurally indistinguishable (see Theorem~\ref{thm:twin} and Proof~\ref{proof:twin}) and an attacker can be certain about the connections of Bob and thus the connection with Alice.}
        \label{fig:twin}
    \end{figure}

    \begin{figure}[!t]
        \centering
        \includegraphics[width=\linewidth]{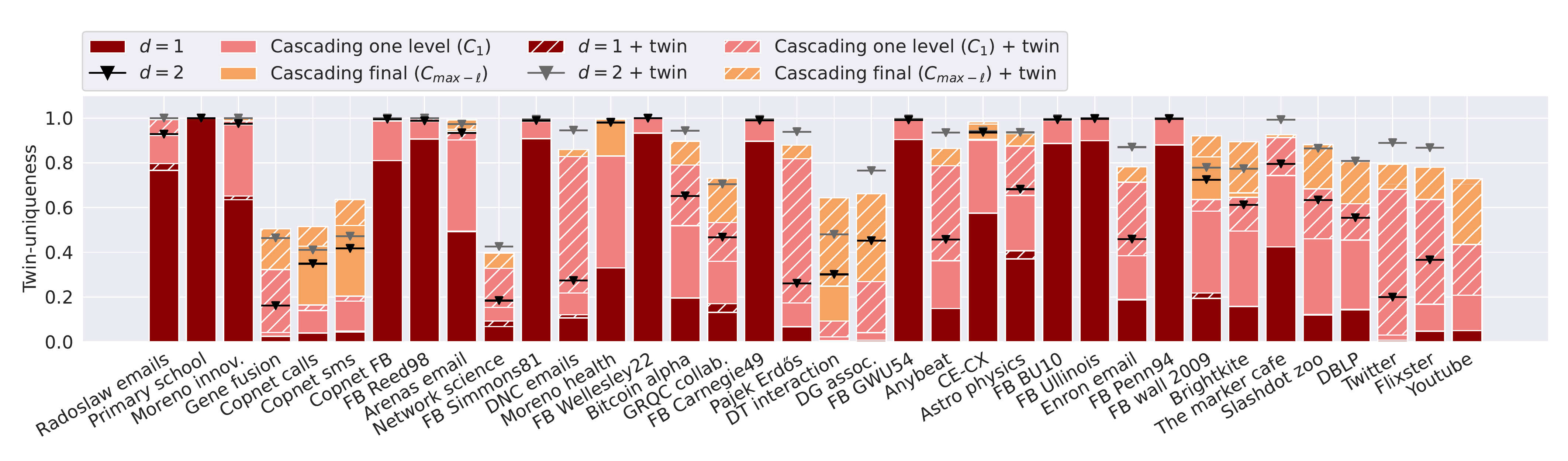}
        \caption{Twin-uniqueness in real-world networks. Fraction of (twin)-unique nodes (vertical axis) on different datasets (horizontal axis) when accounting for different levels of information: 1-neighborhood (red), 2-neighborhood (black and grey line with triangle), one level of cascading (pink) and all levels of cascading (yellow). Results without hatching correspond to results in Fig.~\ref{fig:cascading}. Results with hatching show the increase achieved by including twin-unique nodes.}
        \label{fig:twin_results}
    \end{figure}

\subsubsection*{Twin nodes in real-world network data}\label{sub:emp3}
Figure~\ref{fig:twin_results} shows the fraction of twin-unique nodes for the considered corpus of real-world networks, extending upon Fig.~\ref{fig:cascading}.
Accounting for twin-unique nodes in one level of \emph{anonymity-cascade}  results in an average absolute increase of 0.16, compared to uniqueness. 
For the ``Twitter'' network the largest difference is observed resulting in a twin-uniqueness of 0.65.
An explanation for this strong effect is in the high fraction of twin nodes in the network (0.801, see Table~\ref{tab:data}).
These types of nodes are likely frequently occurring due to the preferential attachment process often observed in these types of social networks.~\cite{barabasi2016network}
This results in many nodes with a low degree, which would be identified using one cascading step if their neighbor is unique. 
Overall, the relative increase is larger if the network has a higher fraction of twin nodes (we detail this relation in~\hyperref[sec:supmat]{Supplementary information}).

For cascading final and \emph{$d$-$k$-anonymity} with $d=2$, we observe similar average increases of 0.21 and 0.18 respectively.
The largest increase is again found for the ``Twitter'' network and equals 0.76	and 0.69.
In contrast, merely accounting for twin nodes using the plain measure of \emph{$1$-$k$-anonymity} almost never increases the fraction; differences of at most 0.04 are observed.
This small effect is likely due to the fact that there are often multiple sets of twin nodes with the same ego network implying that these nodes often have more frequently occurring, likely simple, ego networks.

\section*{Discussion}\label{sec:con}

In this work we discussed new measures and algorithms for anonymity which each correspond to a different attacker scenario, i.e., amount of knowledge of the network structure, and vary in strictness and required computational effort.

First, we explored \emph{$d$-$k$-anonymity} and found that compared to merely measuring ego network uniqueness, information about 2-neighborhoods drastically decreased the anonymity of nodes in both graph models and a wide range of real-world networks. 
However, perfect knowledge of the $2$-neighborhood might in some cases be an unrealistically large amount of knowledge, especially if the ego network is dense. 
Therefore, we extended the measure of \emph{$d$-$k$-anonymity} by means of a cascading step, resulting in \emph{anonymity-cascade}, which models the scenario in which a possible attacker has knowledge about ego networks of two linked nodes, of which one is unique.
\emph{Anonymity-cascade} adds to existing measures in two ways.
First, it accounts for the scenario where a possible attacker has more information than the ego network, but less than the 2-neighborhood.
Second, since \emph{anonymity-cascade} is far less expensive to compute, it can effectively measure anonymity on larger networks with millions of nodes and over ten million edges in minutes.

When assuming above mentioned additional knowledge of one link, we observed major increases in the uniqueness (and thus decreases in anonymity) of virtually all of the considered real-world network datasets.
For 19 of the investigated networks, using one level of cascading more than doubled the fraction of unique nodes.
When continuing this cascading effect, the anonymity decreased further for some networks, but the effect of the first link remains the largest.
Based on results using \emph{$d$-$k$-anonymity} and \emph{anonymity-cascade}, we argue the de-anonymizing effects can not be ignored, and that, as opposed to what was assumed in previous work, information beyond the ego network should be taken into account when measuring anonymity of individuals in networks.

Lastly, we discussed that if all candidates for a node are connected to the exact same nodes, i.e., are twin nodes, the same knowledge can be obtained as when a node is uniquely identified, without any additional attacker knowledge.
Taking into account the notion of twin-uniqueness, which includes all nodes for which all candidates are twin, we observed a large decrease in anonymity for many of the real-world networks.

With this work, we aim to emphasize the importance of research on anonymity measures for networks. 
But moreover, we wish to stress the relevance of balancing the trade-off between attacker knowledge, measure strictness and computation time. 
While we have now arrived at an approach that can process large networks with millions of nodes and edges in a reasonable amount of time, there are still many directions for possible future work.
One avenue for future work could be to explore variants of the cascading measure, using a different starting knowledge than the uniqueness of the ego network. 
Another avenue of research could be to investigate how other graph properties can be exploited, such as distance between candidate nodes, or their membership of network communities.
Another way to extend the work would be to include other properties such as node or edge labels, weights or timestamps.
In the end, the aim is to create feasible measures for anonymity that take into account realistic attacker scenarios, while remaining computable in a realistic amount of time.
Moreover, based on all of these measures, network anonymization methods can be created and applied, ultimately allowing researchers and practitioners to share network data with the assurance that the privacy of individuals in it is maximally guaranteed. 
\section*{Methods}\label{sec:methods}

In this section, we discuss relevant definitions, theorems and proofs as well as a detailed description of the proposed algorithms, ending with a description of the overall experimental setup. 

\subsection*{Notation and definitions}\label{sub:notation}
We define a network or graph $G = (V, E)$ as a set of nodes $V$ and set of edges $\{u, v\} \in E$, where $u, v \in V$.
The set of all nodes in a particular graph $G$ is denoted $V(G)$.
Given two nodes $v, w \in V$ we define the distance $d(v, w)$ as the minimum number of adjacent edges that must be traversed to from node $v$ reach~$w$.
If there is no such path, $d(v, w) = \infty$.
It follows that $d(v, v) = 0$, and since the graph is undirected, $d(v, w) = d(w, v)$.
For a node $v$, we define the $d$-neighborhood $N_d(v)$ as the graph containing all nodes that are at most distance $d$ from $v$, and the set of all edges between these nodes.

\subsection*{The $d$-$k$-anonymity algorithm}\label{sub:dkano}

To measure anonymity, we use $d$-$k$-anonymity,~\cite{dejong2023algorithms} an existing approach that uses isomorphism. We follow notation and definitions used in ~\cite{dejong2023algorithms}.
This measure takes as input a graph and a value for the neighborhood distance $d$, and outputs an equivalence partition of the nodes such that in each equivalence class, all nodes are equivalent to each other.
Based on this partition, the fraction of nodes that are $k$-anonymous can be determined for any given value $k$.
Recall that a node is $k$ anonymous if it is equivalent to $k-1$ nodes. 
If $k=1$ for a node, it is unique.
If a node is in an equivalence class of size $k$, i.e., it is equivalent to $k-1$ nodes, we say that it is $k$-anonymous.

\begin{definition}{Graph isomorphism.}\label{def:iso}
Given two graphs $G = (V, E)$ and $G' = (V', E')$, a graph isomorphism is a bijective function $\phi : V \rightarrow V'$ such that for each $v, w \in V$ it holds that $\{\phi(v), \phi(w)\} \in E'$ precisely when $\{v, w\} \in E$.
\end{definition}

We call two graphs isomorphic if there is at least one isomorphism between them.
A special form of isomorphism is an automorphism, denoted by $\gamma$. This is an isomorphism from a graph onto itself.
If there exists an automorphism such that two nodes $v, w \in V$ are mapped onto each other, they are said to be in the same orbit.

\begin{definition}{$d$-Equivalence.}\label{def:equivalent}
Two nodes $v, w \in V$ are $d$-equivalent
if:
1) their d-neighborhoods are isomorphic
and 2) there is an isomorphism $\phi$ such that $\phi(v)=w$.
\end{definition}

If two nodes are $d$-equivalent, they are equivalent with respect to \emph{$d$-$k$-anonymity}. 
It is not difficult to show that $d$-equivalence indeed satisfies the properties of an equivalence relation~\cite{loo2022topological} (identity, reflexivity and transitivity). 
It can also be demonstrated that if two nodes are $d+1$-equivalent, they must also be $d$-equivalent. 
In previous work~\cite{dejong2023algorithms} several methods for efficient calculation of this measure, such as the use of a cache and filtering out nodes based on graph invariants, are discussed.
Code for computing $d$-$k$-anonymity can be found at~\url{https://github.com/RacheldeJong/dkAnonymity}.

\subsection*{The anonymity-cascade algorithm}\label{sub:casc}

Below we describe the workings of \emph{anonymity-cascade}, the algorithm proposed in this paper that extends \emph{$d$-$k$-anonymity} with $\ell \geq 1$ cascading steps. 
It takes as input a graph and depth parameter $\ell$, and outputs the uniquely identified nodes.

\begin{itemize}
  \setlength\itemsep{0em}
  \setlength{\parskip}{0pt}
    \item \textbf{Input: } Graph $G = (V, E)$,  maximum cascading level $\ell$ 
    \item $U_{cur} =  U = $ all nodes with a unique ego network (obtained using \emph{$d$-$k$-anonymity} with $k=1$ and $d=1$)
    \item $\ell_{cur} = 1$
    \item While $\ell_{cur} \leq \ell$ and $U_{cur} \neq \emptyset$:
    \begin{itemize}
    \setlength\itemsep{0em}
    \setlength{\parskip}{0pt}
        \item $U_{new} = \emptyset$
        \item For each node $u \in U_{cur}$:
        \begin{itemize}
            \setlength\itemsep{0em}
            \setlength{\parskip}{0pt}
            \item For each $v \in V(N_1(u)) - \{u\}$:
            \begin{itemize}
                \setlength\itemsep{0em}
                \setlength{\parskip}{0pt}
                \item Check if there is a node $v' \in V(N_1(u)) - \{u, v\}$ such that $v$ and $v'$ are $1$-equivalent
                \item If there is no such node: $U_{new} = U_{new} \cup \{v\}$
            \end{itemize}
        \end{itemize}
        \item $U_{cur} = U_{new} - U$
        \item $U = U \cup U_{new}$ 
        \item $\ell_{cur} = \ell_{cur} + 1$
    \end{itemize}
    \item \textbf{Output: } Set of uniquely identified nodes $U$
\end{itemize}

Interestingly, \emph{anonymity-cascade} has the useful property that all nodes that are unique for $C_1$, are also unique for \emph{$2$-$k$-anonymity} and hence \emph{$2$-$k$-anonymity} is at least as strict as $C_1$ (see Theorem~\ref{thm:casc} and Proof~\ref{proof:casc}).
We use a proof by contradiction and first assume that there is a pair of nodes $v, w$ that is $2$-equivalent and $v$ is unique for $C_1$.
We then show a unique neighborhood is contained in the 2-neighborhood of both $v$ and $w$, which contradicts the assumption that $v$ is unique for $C_1$.
Code for computing \emph{anonymity-cascade} can be found at~\url{https://github.com/RacheldeJong/Anonymitycascade}.

\begin{theorem}\label{thm:casc}
    In a graph $G = (V, E)$, a node $v$ that is unique using anonymity-cascade with $\ell = 1$ ($C_1$) is also unique using $2$-$k$-anonymity.
\end{theorem}

\begin{proof}\label{proof:casc}
    Node $v$ is connected to a node $u$ where $N_1(u)$ is unique.
    It holds that $N_1(u)$ is contained in $N_2(v)$, as $d(v, u') \leq 2$ for all $u' \in V(N_1(u))$.
    Let us now assume there is a node $w$ which is 2-equivalent to $v$.
    This implies that $N_2(v)$ and $N_2(w)$ are isomorphic.
    Hence, $N_1(u)$ should occur in $N_2(w)$.
    Since $N_1(u)$ is unique, this implies $w$ should be a neighbor of $u$.
    However, $v$ and $w$ can only be 2-equivalent if they are 1-equivalent.
    This contradicts the assumption that $v$ is unique using $C_1$.
    
\end{proof}

\subsection*{Twin node preprocessing}\label{sub:twin}
For the definition of twin nodes, we distinguish between a closed neighborhood $N_d(v)$ as the neighborhood defined in~\hyperref[sub:notation]{Notation and definitions} and the open neighborhood of a node $N_d'(v) = N_d(v) - \{v\}$.

\begin{definition}{Twin nodes.}\label{def:twin}
    Given a graph $G=(V, E)$, nodes $v_1\neq v_2 \in V$ are closed twin nodes if
    $N_1(v_1) = N_1(v_2) $ or open twin nodes if $N_1'(v_1) = N_1'(v_2)$.
\end{definition}

Twin nodes are in the same orbit, which we show by constructing an automorphism that maps the twin nodes onto each other, and all other nodes onto themselves.
The proof used is similar to the proof presented in previous work.~\cite{dejong2023algorithms}

\begin{theorem}\label{thm:twin}
	Given a graph $G  =(V, E)$ and nodes $v_1, v_2 \in V$. If nodes $v_1, v_2$ are twin nodes, then they are in the same orbit.
\end{theorem}

\begin{proof}\label{proof:twin}
    Given twin nodes $v_1$ and $v_2$, we define an automorphism $\gamma: V\to V$ that swaps $v_1$ and $v_2$ and maps all other nodes onto themselves. 
    To show that $\gamma$ is a valid automorphism, note that for any $\{v_1, w \neq v_2\}\in E$, we have $\{\gamma(v_1),\gamma(w)\} = \{v_2, w\} \in E$ because of the twin node property (Definition~\ref{def:twin}).
    Similarly we have $\{\gamma(v_2), \gamma(w)\}= \{v_1, w\} \in E$ for all $\{v_2,w\neq v_1\}\in E$.
    If $\{v_1, v_2\} \in E$, then $\{\gamma(v_1), \gamma(v_2)\} = \{v_1, v_2\} \in E$.
    Hence, $v_1$ and $v_2$ are in the same orbit. 
\end{proof}

We find twin nodes using the algorithm below, which serves as a preprocessing step before \emph{$d$-$k$-anonymity} or \emph{anonymity-cascade} is executed. 
It is worth to note that a node can be a twin of more than one node, and hence a set of twin nodes can have a size larger than two.

\begin{itemize}
  \setlength\itemsep{0em}
    \item \textbf{Input: } Graph $G = (V, E)$
    \item Create two dictionaries $M_o$ and $M_c$, used to map neighborhoods onto nodes with this resp. open and closed neighborhood
    \item For each node $v \in V$:
    \begin{itemize}
        \item If $V(N_1'(v)) \in M_o$: $M_o[V(N_1'(v))] = M_o[V(N_1'(v))] \cup \{v\}$
        \item Else if $V(N_1(v))  \in M_c$: $M_c[V(N_1(v))] = M_c[V(N_1(v))] \cup \{v\}$
        \item Else: $M_o[V(N_1'(v))] = \{v\}$ and $M_c[V(N_1(v))] = \{v\}$
    \end{itemize}
    \item \textbf{Output: } Sets of open and closed twin nodes $M_o$, $M_c$
\end{itemize}

After finding twin nodes, they are taken into account as follows.
When computing \emph{$d$-$k$-anonymity}, we select one node for each set of twin nodes that has to be taken into account when computing anonymity.
All other twin nodes are not taken into account during computation, and afterwards added to their corresponding equivalence classes.~\cite{dejong2023algorithms}
When computing twin-uniqueness for both \emph{$d$-$k$-anonymity} and \emph{anonymity-cascade}, both unique nodes and nodes for which all candidates are twins are twin-unique.
For \emph{anonymity-cascade}, we additionally reuse twin-unique nodes to continue the cascading effect.

\subsection*{Experimental setup}\label{sub:exp}
The results for \emph{$d$-$k$-anonymity} and \emph{anonymity-cascade} were obtained using the default settings of the code repositories linked above. 
When computing twin-uniqueness, for both \emph{$d$-$k$-anonymity} and \emph{anonymity-cascade} we compute and take into account both open and closed twin nodes; these computation times are not included in reported runtimes for \emph{$d$-$k$-anonymity}, and are included for runtimes for \emph{anonymity-cascade}.

The graph models used for the experiments consist of the Erdős–Rényi,~\cite{erdos1960evolution} Barabási-Albert,~\cite{barabasi1999emergence} and Watts-Strogatz~\cite{watts1998collective} models generated using NetworkX.~\cite{networkx}
For the WS graphs we use random wiring probability $p_r=0.5$, similar to previous work.~\cite{romanini2020privacy}
All results reported on graph models are averaged over ten networks.

For experiments on real-world networks, self loops and nodes without edges are removed from the original network dataset. 
Moreover, edge weights, timestamps and directionality are ignored. 
Reported runtimes are averaged over five runs.
All experiments are conducted on a machine with 1TB RAM, 64 AMD EPYC 7601 cores, and 128 threads. 
During the experiments, each run uses one thread, which is not shared with other processes.

\section*{Data availability}
All network datasets are available in the repositories cited in Table~\ref{tab:data}.

\bibliography{bibliography}

\section*{Acknowledgements}

This research was made possible by the Platform Digital Infrastructure SSH (http://www.pdi-ssh.nl). We would also like to thank the POPNET team (https://www.popnet.io) for various helpful suggestions and discussions.

\section*{Author contributions}

R.J., M.L., and F.T. conceptualized the study and methodology. R.J. wrote the manuscript. M.L. and F.T.
acquired funding, supervised and administered the project. R.J. developed the software, performed
the experiments and visualized the results. All authors reviewed and approved the manuscript.

\section*{Competing interests}
The authors declare no competing interests.

\newpage
\section*{Supplementary information}\label{sec:supmat}

\subsection*{$d$-$k$-Anonymity for  $d=3$, $d=4$ and $d=5$}
Supplementary Figure~\ref{fig:modelsdkcomplete} shows additional results for Section  ``Beyond the ego network in graph models'' of the main manuscript.
The figure is similar to Fig. 2, showing the fraction of unique nodes in graph models with knowledge of the $d$-neighborhood, for $d=3$, $d=4$ and $d=5$.
The figures show that  further increasing the distance beyond a value 2 has a minimal de-anonymizing effect for these graph models; the largest effect is observed when moving from distance 1 to distance 2, as shown in Fig. 2.

Supplementary Figure~\ref{fig:dkcomplete} shows additional results for Section ``Beyond the ego network in real-world network data'' of the main manuscript, including uniqueness for $d$-$k$-anonymity with $d=1$ up to $d=5$.
The figure is similar to Fig. 5, and includes results shown in Fig. 3 for $k=1$.
The figure shows that information beyond $d=2$ has only a small to no de-anonymizing effect on these networks.

    \begin{figure}[htbp]
            \centering
            \includegraphics[width=\textwidth]{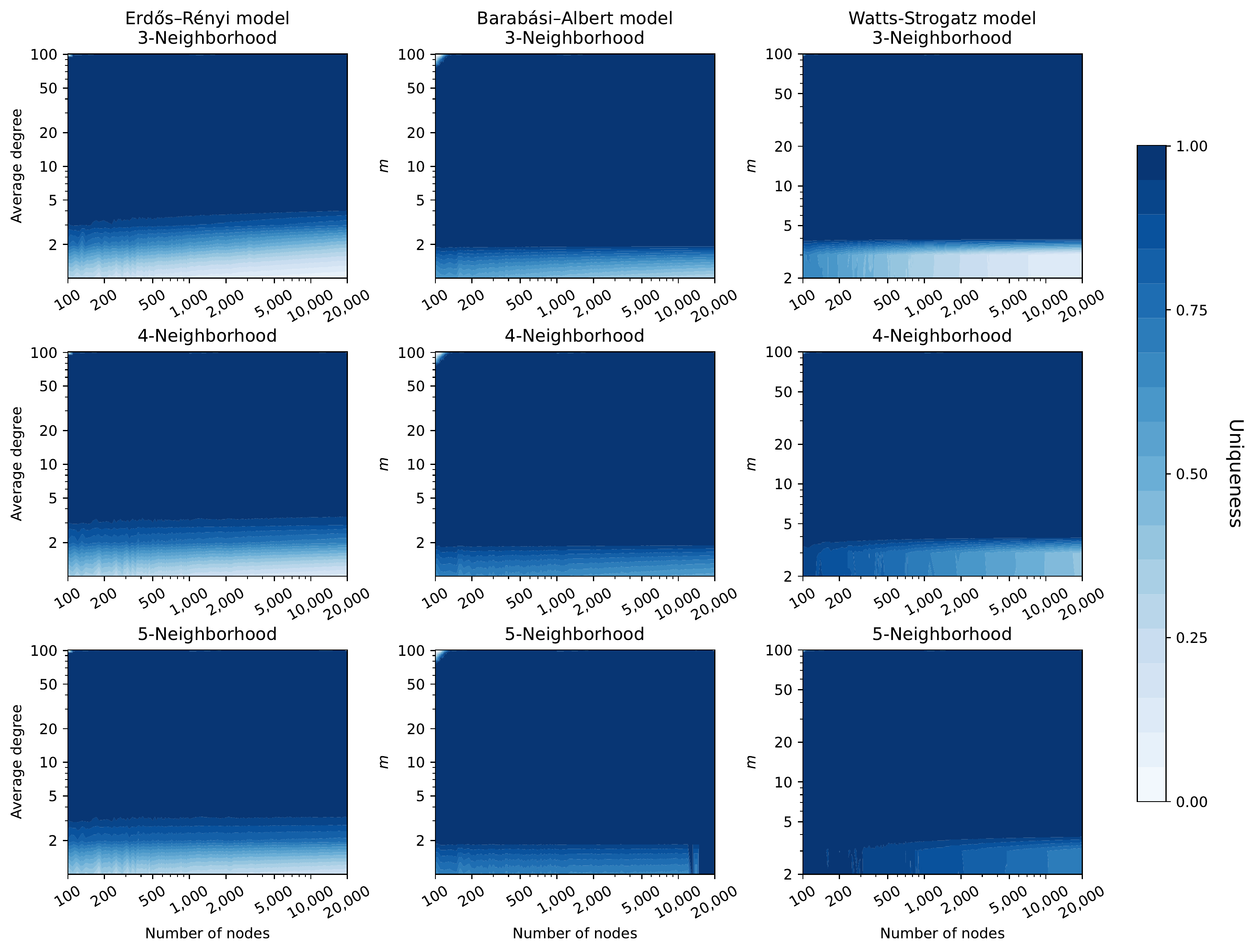}
            \caption{Uniqueness maps using $d$-$k$-anonymity. Maps show network uniqueness, indicated by color, when using information of the $3$-neighborhood (top row), $4$-neighborhood (middle row) and $5$-neighborhood. Results are shown for the Erdős–Rényi (left), Barabási–Albert (middle) and Watts–Strogatz (right) model with different sizes (horizontal axis) and average degree or \added{ $m$, an } equivalent thereof (vertical axis).}
           \label{fig:modelsdkcomplete}
    \end{figure}

    \begin{figure}[htbp]
            \centering
            \includegraphics[width=\textwidth]{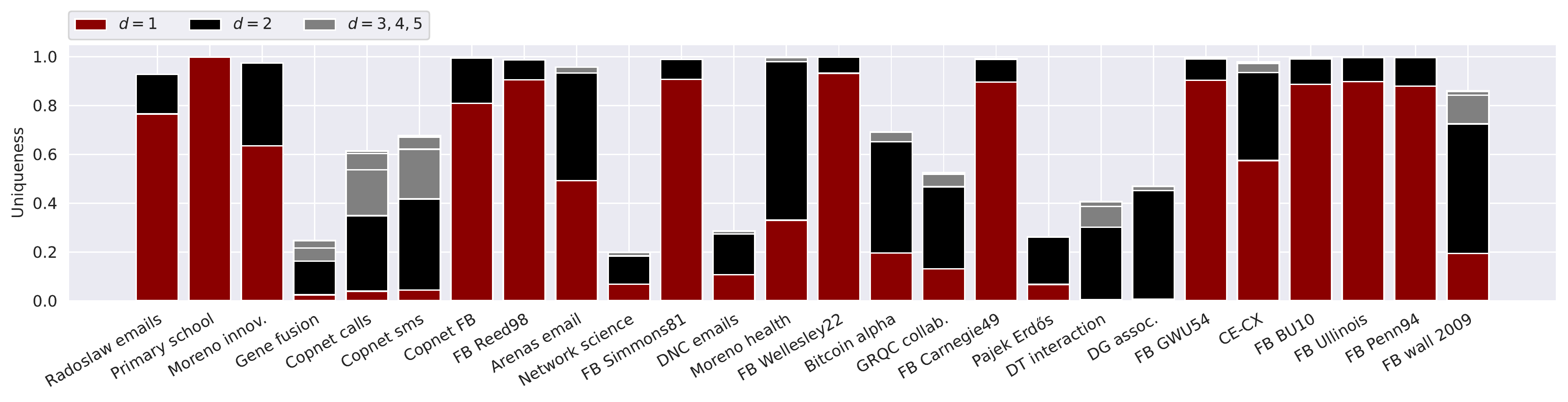}
            \caption{Uniqueness in real-world networks for $d$-$k$-Anonymity with $d=1$ up to $d=5$. Fraction of unique nodes (vertical axis) on different datasets (horizontal axis) when accounting for different levels of information: 1-neighborhood (red), 2-neighborhood (black), 3- 4- and 5-neighborhood (grey).}
           \label{fig:dkcomplete}
    \end{figure}
\newpage

\subsection*{$d$-$k$-Anonymity in dense graph models}
\added{Similarly to the uniqueness maps in Fig. 2 and 4 of the main manuscript, we show the uniqueness of graph models with a density up to their maximum degree in Supplementary Fig.~\ref{fig:modelsdense}.
A step size of 0.01 is used, and all results are averaged over 10 runs.}

\added{Overall, we see that for most of the graph models, for higher densities all nodes become unique.
We observe non-uniqueness for graphs with either low densities, or for very high densities.
The graphs with uniqueness at high densities are (near)-complete graphs, and can only be clearly distinguished in the top left of these figures for small graphs.}

    \begin{figure}[htbp]
            \centering
            \includegraphics[width=\textwidth]{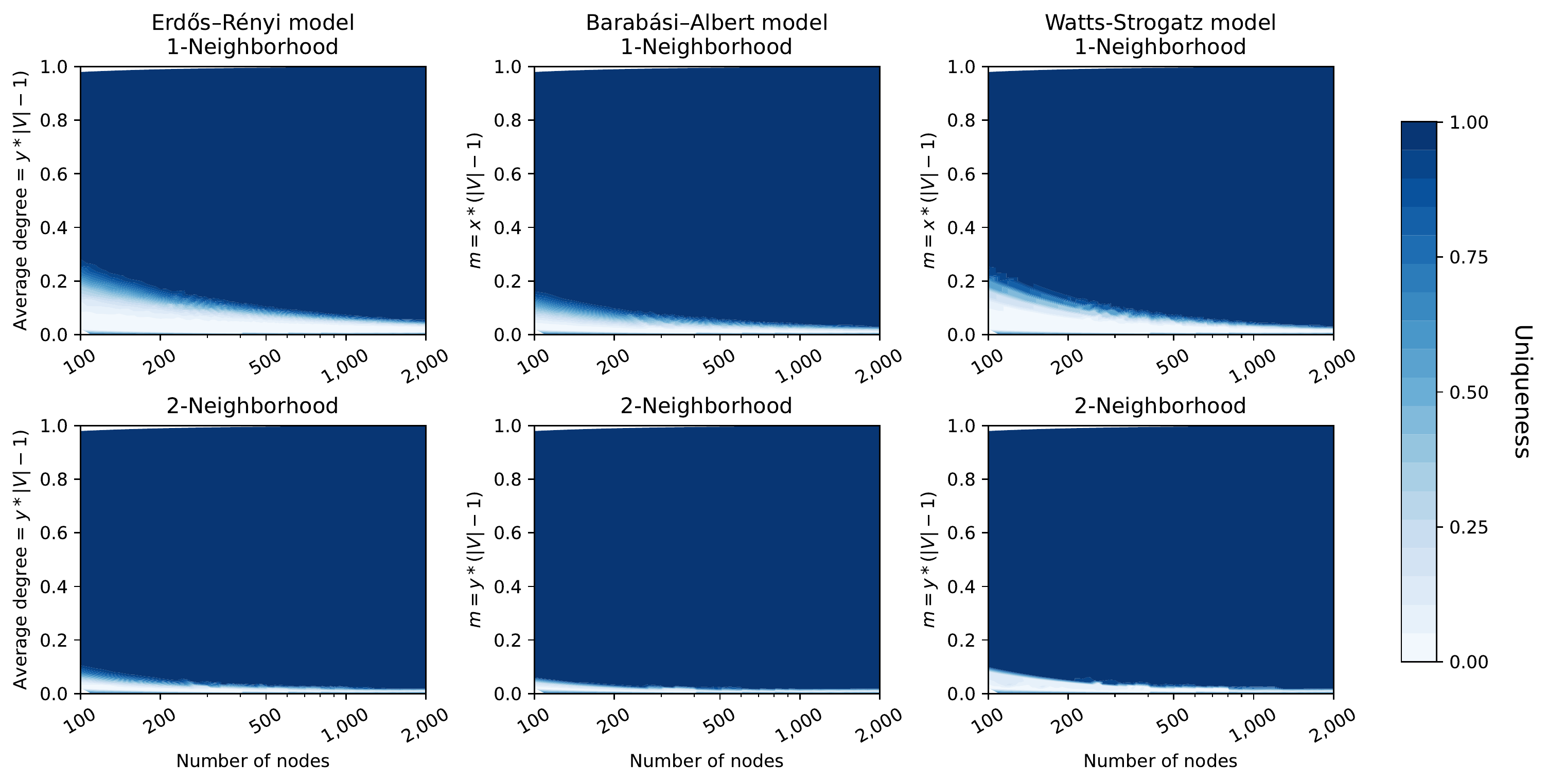}
            \caption{\added{Uniqueness maps using $d$-$k$-anonymity in dense graph models. Maps show network uniqueness, indicated by color, when using information of the $1$-neighborhood (top row) and $2$-neighborhood (bottom row). Results are shown for the Erdős–Rényi (left), Barabási–Albert (middle) and Watts–Strogatz (right) model with different sizes (horizontal axis) and average degree or $m$, an equivalent thereof (vertical axis).}}
           \label{fig:modelsdense}
    \end{figure}

\newpage
\subsection*{Real-world networks vs. graph models}
\added{To compare the results on real-world networks with those obtained for the graph models, Supplementary Fig. ~\ref{fig:models_exp} shows for each network the position in the uniqueness map, similar to Figures 2 and 4 of the main manuscript, and indicates the uniqueness by color.
Similar to the results obtained by Romanini \emph{et al.}~\cite{romanini2020privacy} and for the graph models, the figure shows that average degree has a large effect on the uniqueness obtained.
The number of nodes has an overall smaller effect.}

    \begin{figure}[htbp]
        \centering
        \includegraphics[width=0.7\textwidth]{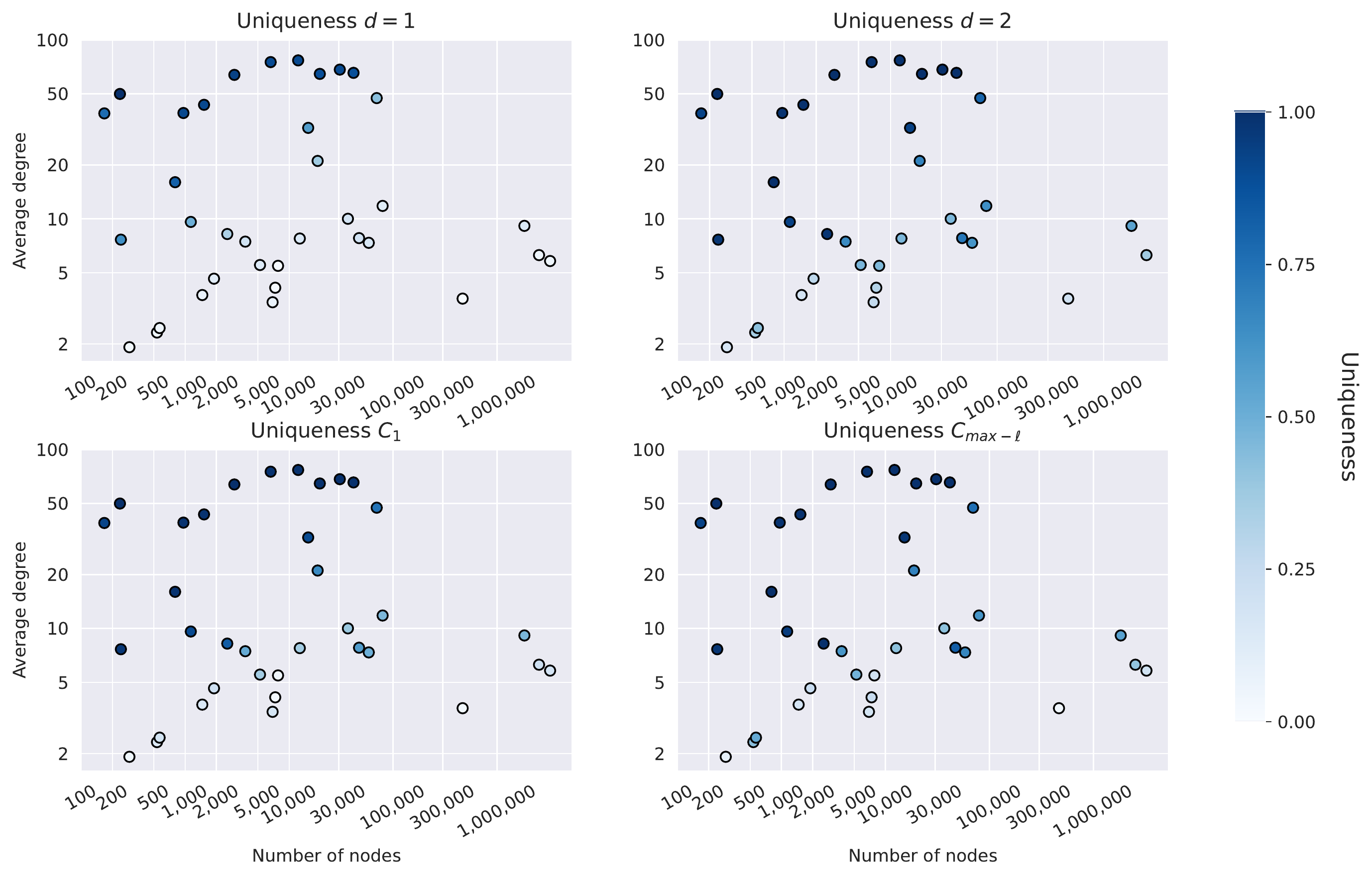}
        \caption{\added{Real-world networks positioned in uniqueness maps. For all networks listed in Table 1, their position in the uniqueness map, which is used in Figures 2 and 4 of the main manuscript, is shown and the uniqueness is indicated based on the color.}}
        \label{fig:models_exp}
    \end{figure}

\subsection*{Graph properties vs. experimental results}

\begin{table}[!b]
\scriptsize
\center
\begin{tabular}{@{}rrlrlrlrl@{}}
\toprule
                    & \multicolumn{2}{c}{Uniqueness $d=1$}                                                                   & \multicolumn{2}{c}{$max-\ell$}                                                                         & \multicolumn{2}{c}{Uniqueness $(d=2) - (d=1)$}                                                    & \multicolumn{2}{c}{Uniqueness $(C_1) - (d=1)$}                                                            \\ \cmidrule(l){2-9} 
Network property            & \multicolumn{1}{l}{\begin{tabular}[c]{@{}r@{}}Pearson \\ correlation\end{tabular}} & p-value           & \multicolumn{1}{r}{\begin{tabular}[c]{@{}r@{}}Pearson \\ correlation\end{tabular}} & p-value           & \multicolumn{1}{c}{\begin{tabular}[r]{@{}r@{}}Pearson \\ correlation\end{tabular}} & p-value  & \multicolumn{1}{r}{\begin{tabular}[r]{@{}r@{}}Pearson \\ correlation\end{tabular}} & p-value           \\ \midrule
Nodes               & -0.113                                                                             & 5.38E-01          & 0.203                                                                              & 2.65E-01          & -0.123                                                                             & 5.03E-01 & 0.385                                                                              & \textbf{2.97E-02} \\
Edges               & 0.305                                                                              & 8.91E-02          & -0.146                                                                             & 4.24E-01          & -0.098                                                                             & 5.93E-01 & 0.063                                                                              & 7.33E-01          \\
Average degree      & \textbf{0.865}                                                                     & \textbf{1.75E-10} & \textbf{-0.574}                                                                    & \textbf{5.85E-04} & -0.068                                                                             & 7.13E-01 & -0.322                                                                             & 7.24E-02          \\
Median degree       & \textbf{0.895}                                                                     & \textbf{4.79E-12} & \textbf{-0.621}                                                                    & \textbf{1.48E-04} & -0.041                                                                             & 8.24E-01 & \textbf{-0.418}                                                                    & \textbf{1.71E-02} \\
Max degree          & \textbf{0.599}                                                                     & \textbf{2.89E-04} & -0.344                                                                             & 5.37E-02          & -0.077                                                                             & 6.75E-01 & -0.266                                                                             & 1.41E-01          \\
Density             & \textbf{0.412}                                                                     & \textbf{1.93E-02} & \textbf{-0.480}                                                                    & \textbf{5.42E-03} & -0.002                                                                             & 9.92E-01 & -0.228                                                                             & 2.10E-01          \\
Transitivity        & 0.007                                                                              & 9.70E-01          & -0.263                                                                             & 1.46E-01          & -0.001                                                                             & 9.96E-01 & -0.035                                                                             & 8.48E-01          \\
Assortativity       & 0.085                                                                              & 6.44E-01          & 0.122                                                                              & 5.06E-01          & 0.122                                                                              & 5.06E-01 & 0.228                                                                              & 2.09E-01          \\
Diameter            & \textbf{-0.599}                                                                    & \textbf{2.92E-04} & \textbf{0.809}                                                                     & \textbf{2.06E-08} & -0.136                                                                             & 4.58E-01 & 0.179                                                                              & 3.28E-01          \\
Average path length & \textbf{-0.740}                                                                    & \textbf{1.28E-06} & \textbf{0.854}                                                                     & \textbf{4.94E-10} & -0.114                                                                             & 5.35E-01 & 0.112                                                                              & 5.41E-01          \\ \bottomrule
\end{tabular}
\caption{\added{Graph properties and correlation with results. Given various graph measures, denoted in the leftmost column, the Pearson correlation is shown with the outcomes (second to last columns). For each outcome both the Pearson correlation and p-value are given. Values with $p<0.05$ and Pearson correlation larger than 0.4, or smaller than -0.4, are shown in bold.}}
\label{tab:corr}
\end{table}

\added{To get a better understanding of the obtained results in real-world networks presented in the ``Beyond the ego network'' and ``Anonymity-cascade'' sections of the main manuscript, we aim to find which graph properties relate to various uniqueness results and experimental outcomes.
The results considered are the uniqueness using $1$-$k$-anonymity, the highest achieved cascading depth $max-\ell$, and the difference obtained in uniqueness by instead of $1$-$k$-anonymity considering $2$-$k$-anonymity, or one level of cascading $C_1$.
Since we have a diverse set of networks from different categories, we compute the Pearson correlation and the corresponding $p$-value in Supplementary Table~\ref{tab:corr} for each combination of graph property (computed using igraph~\cite{igraph}) and outcome.
Results are computed for the networks in Table 1 of the main manuscript, excluding the four largest networks due to large runtimes for computing several of the properties.}

\added{The results in Supplementary Table~\ref{tab:corr} show that there are correlations between various graph properties and both the uniqueness at $d=1$ and $max-\ell$.
For uniqueness, the average, median and maximum degree and density are all positively correlated.
Overall, these properties are very related to (or to some extent equivalent to) density, which corresponds to the results obtained by Romanini \emph{et al.}:~\cite{romanini2020privacy} networks with a higher density overall have a higher uniqueness.
This relation can also be seen for graph models in the uniqueness shown in Fig. 2 of the main manuscript.
The table also shows that diameter and average path length are negatively correlated to the uniqueness.
For $max-\ell$, the reverse holds. This is negatively correlated with density, and positively with diameter and average path length.}

\added{For the difference in uniqueness when looking beyond the neighborhood, we only found two substantial correlations when using one level of cascading.
Overall a higher maximum degree is is negatively correlated with a larger effect of one level of cascading.
For all combinations with a substantial correlation we include a plot with regression line in Supplementary Fig.~\ref{fig:properties_corr}.}

    \begin{figure}[b!]
        \centering
        \includegraphics[width=.89\textwidth]{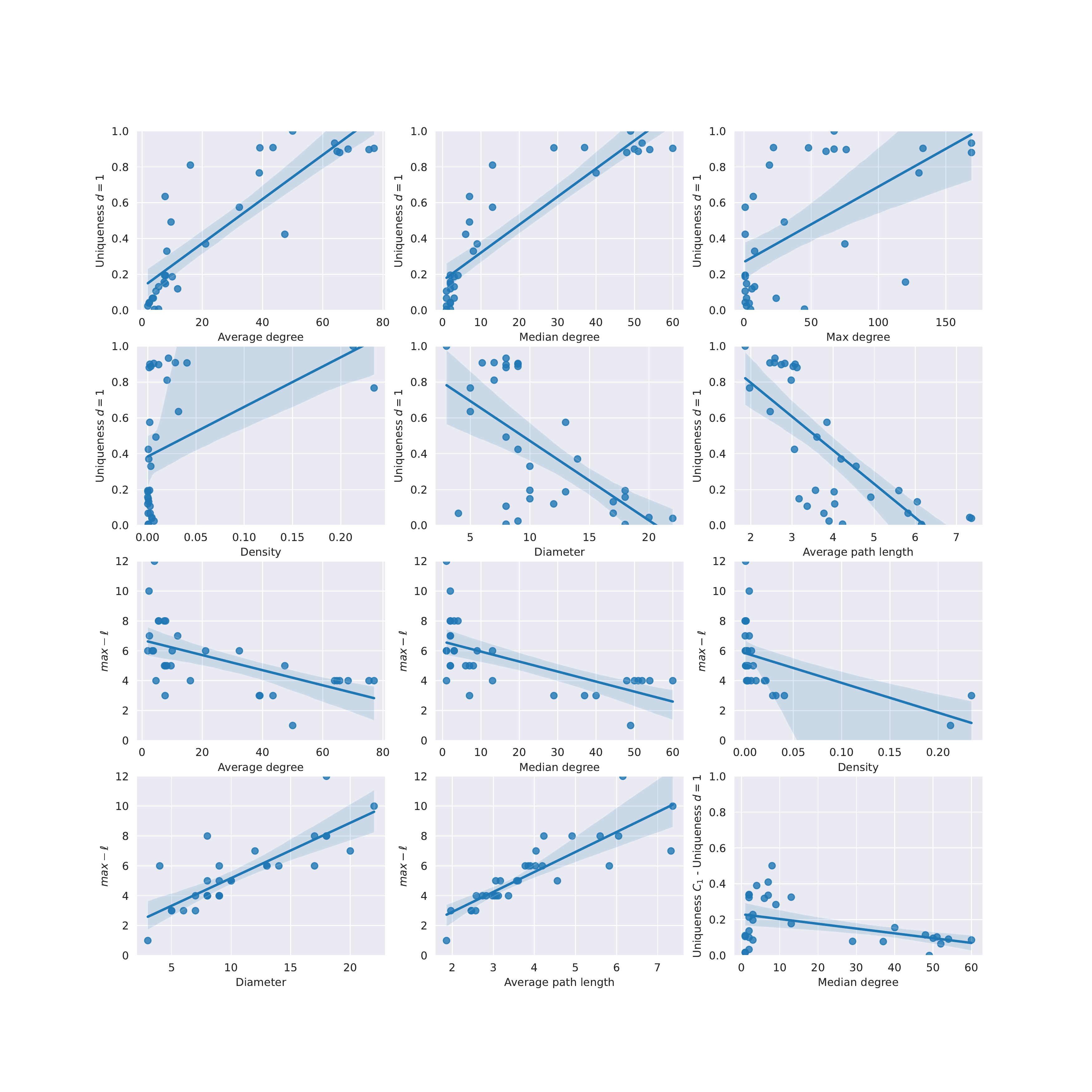}
        \caption{\added{Graph properties and outcomes. Each vertical axis denotes a particular experimental outcome (listed in the first row of Supplementary Table~\ref{tab:corr}), each horizontal axis a graph property (listed in the first column of Supplementary Table~\ref{tab:corr}). Only combinations with $p<0.05$ and Pearson correlation larger than 0.4, or smaller than -0.4 are included.}}
        \label{fig:properties_corr}
    \end{figure}

\newpage
\subsection*{Effect at different levels of anonymity-cascade}

The following figures are supplementary to Section ``Anonymity-cascade'' of the main manuscript and give additional insights into the effect on anonymity at different levels of \emph{anonymity-cascade}.
Supplementary Figure~\ref{fig:casclevels} shows the largest level of cascading for the graph models.
For these results, \emph{anonymity-cascade} is used for each graph model, the highest level ($max-\ell$) is measured and the average over 10 runs is reported.
This overall shows that ($max-\ell$) increases as the graph size increases. 
In some more sparse graphs with average degree or $m$ around 2 (ER, BA) or 5 (WS), these levels achieve the highest values.

Supplementary Figure~\ref{fig:levels} shows how many nodes are uniquely identified at each level of cascading in the real-world networks.
Overall, the results show that the largest de-anonymizing effect occurs at lower levels of cascading.
The cascading step can continue for many levels, but with a much smaller de-anonymizing effect.

    \begin{figure}[htbp]
        \centering
        \includegraphics[width=\textwidth]{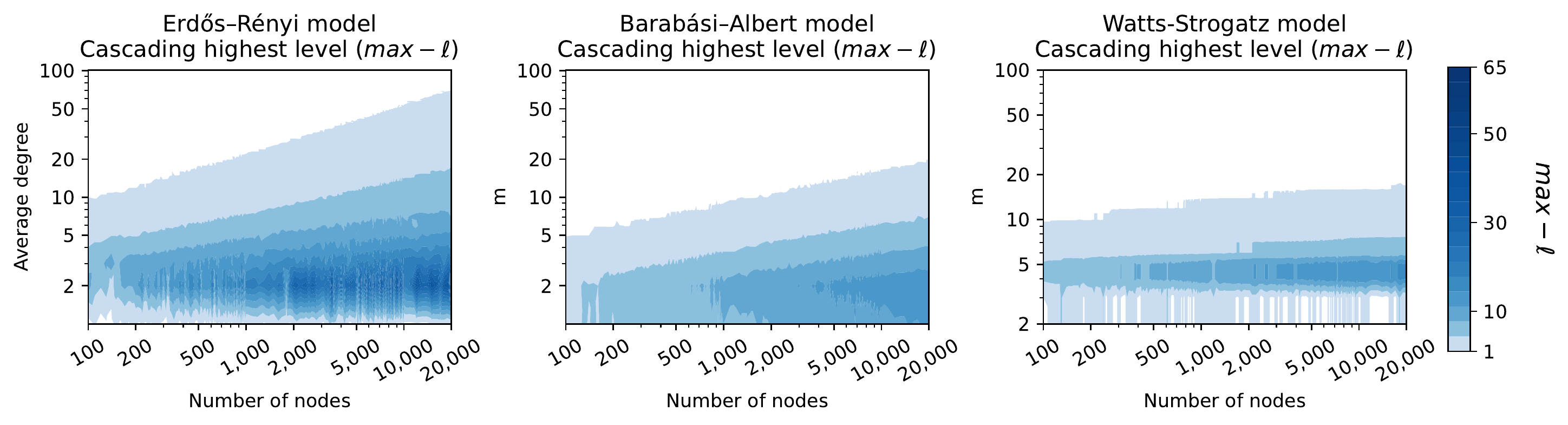}
        \caption{Cascading levels in graph models. Results for graph models showing the highest level achieved ($max-\ell$) by the cascading algorithm for ER (left) BA (middle) and WS (right) graphs. All results are averaged over 10 generated graphs.}
        \label{fig:casclevels}
    \end{figure}
    
    \begin{figure}[htbp]
        \centering
        \includegraphics[width=0.4\linewidth]{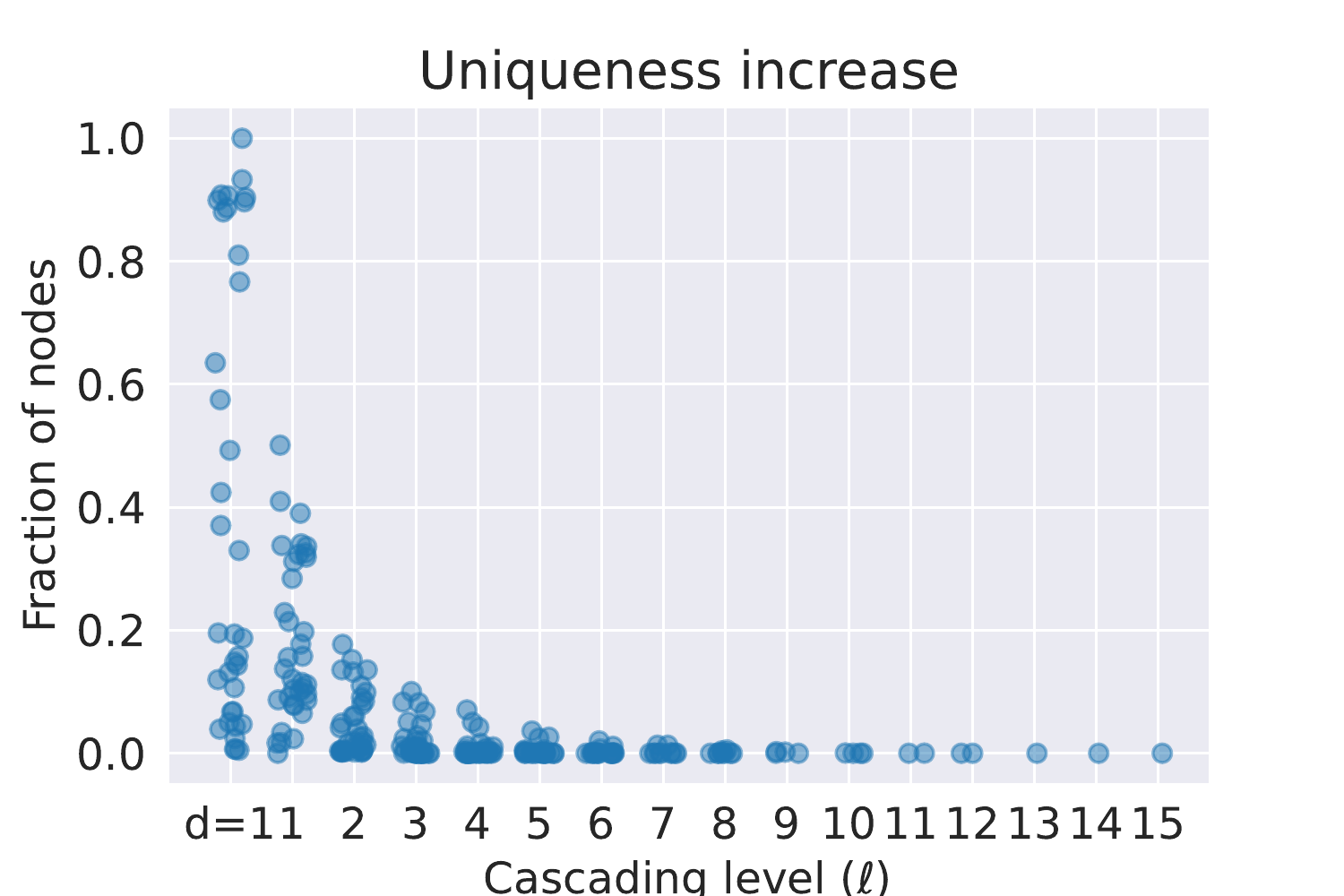}
        \includegraphics[width=0.4\linewidth]{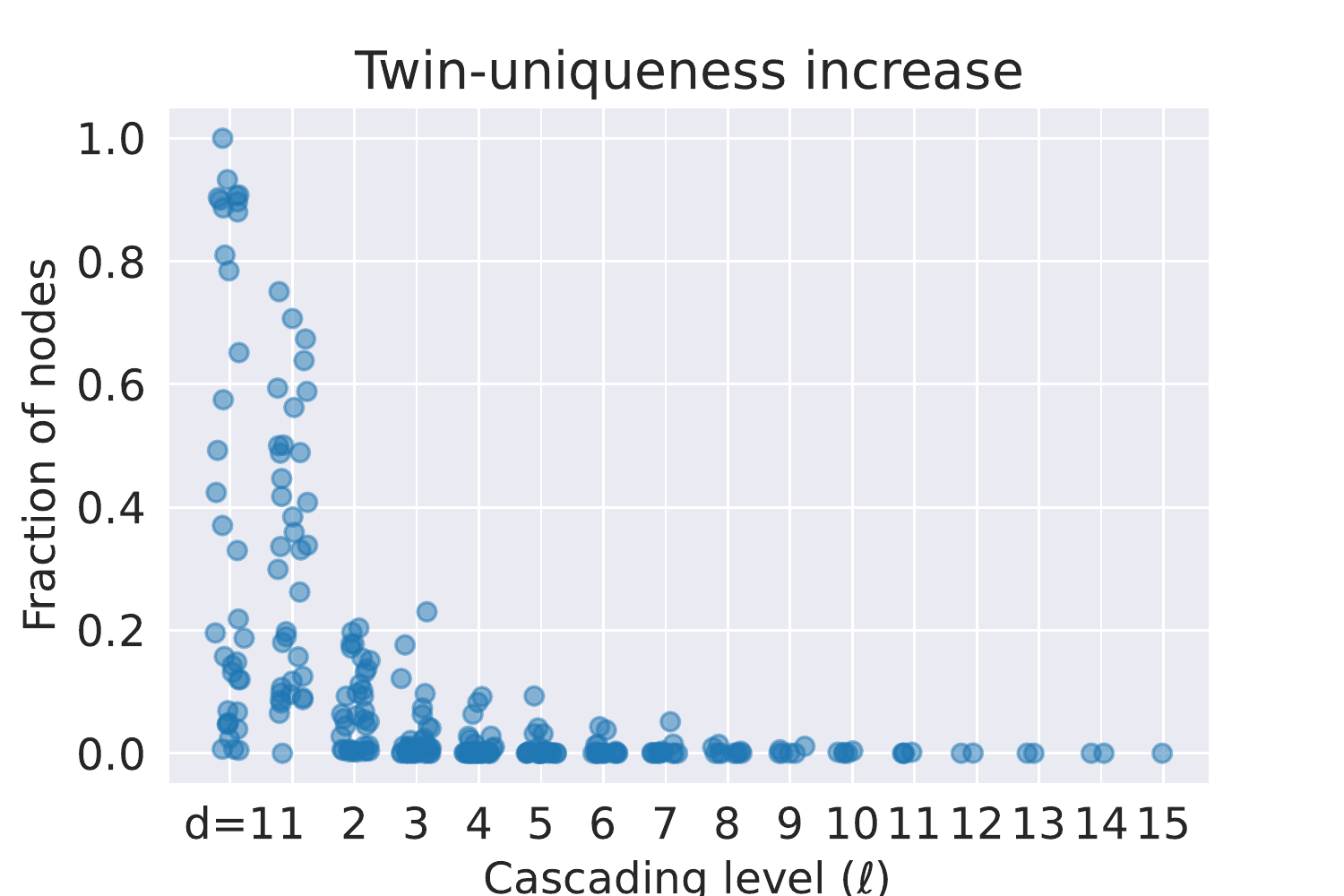}
        \caption{Cascading effect per level. Fraction of identified nodes (vertical axis) at $d=1$ and each subsequent level of cascading when accounting for uniqueness (left) and twin-uniqueness (right). Results summarized over all networks in Table 1 of the main manuscript.}
        \label{fig:levels}
    \end{figure}

\newpage
\subsection*{Twin node effect}
Supplementary Figure~\ref{fig:twins} shows the effect of twin nodes and how it influences twin-uniqueness vs. uniqueness.
When a network contains a higher fraction of twin nodes, the overall effect is larger.

    \begin{figure}[htbp]
        \centering
        \includegraphics[width=0.4\textwidth]{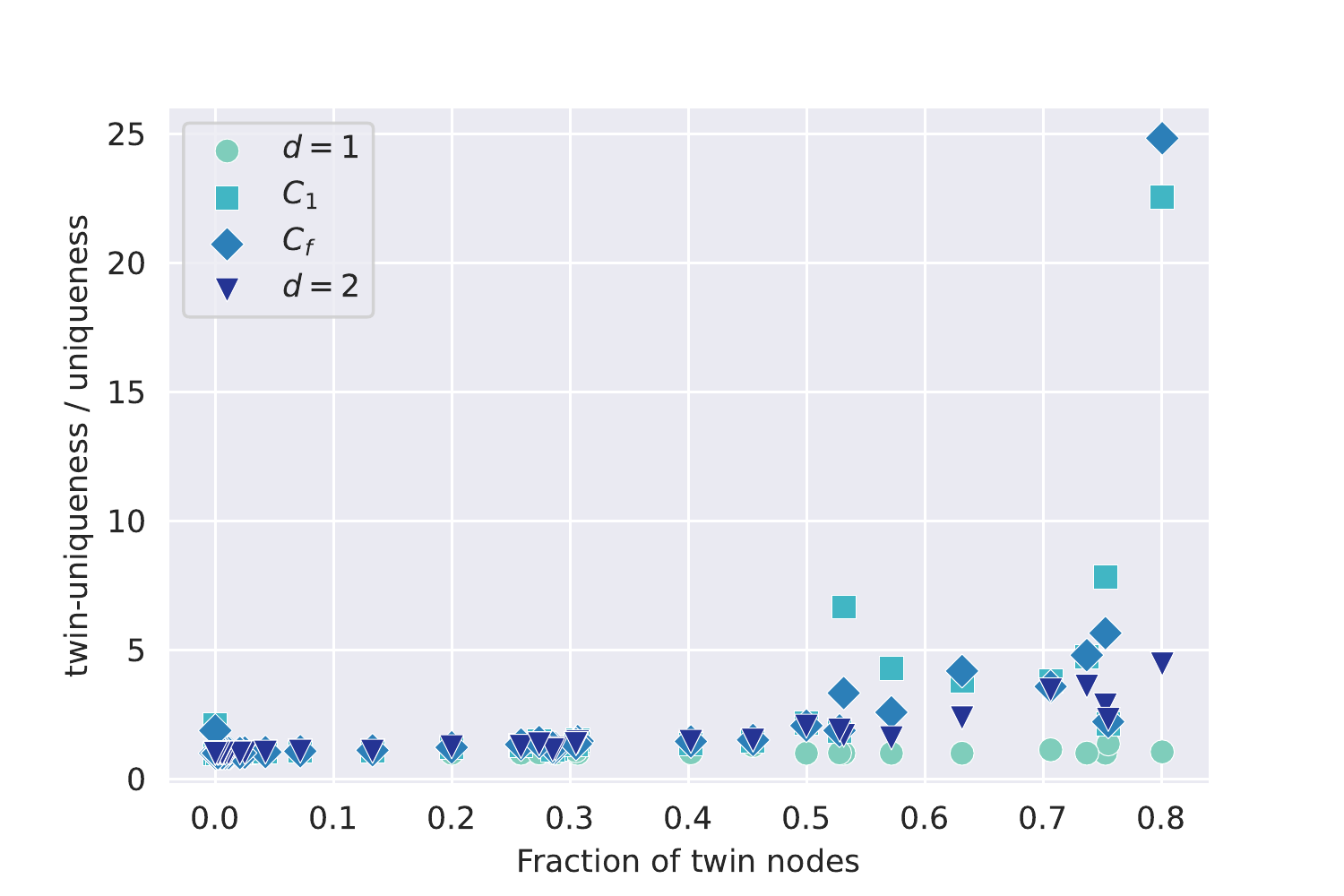}
        \caption{Effect of twin nodes on twin-uniqueness. Relation between the fraction of twin nodes in the network (horizontal axis) and how much the uniqueness increases (vertical axis) compared to when twin nodes are not taken into account.}
        \label{fig:twins}
    \end{figure}

\end{document}